\documentclass[journal,onecolumn]{IEEEtran}
\usepackage{bbding}
\usepackage{pifont}
\usepackage{amsfonts}
\usepackage{mathrsfs}
\usepackage{amsthm}
\usepackage{bm}
\usepackage{graphicx}
\usepackage{epstopdf}
\usepackage[tbtags]{amsmath}
\usepackage{mathtools}
\usepackage{color}
\usepackage{multirow}
\usepackage{amssymb}
\usepackage{nicematrix}
\usepackage{booktabs}
\usepackage{bookmark}
\usepackage{enumerate}
\usepackage{tabularray}
\usepackage{ninecolors}
\usepackage{float}
\usepackage{subcaption}
\usepackage[numbers,sort&compress]{natbib}
\graphicspath{{pics/}{profile_photos/}}
\newtheorem{theorem}{Theorem}

\newtheorem{lemma}{Lemma}
\newtheorem{corollary}{Corollary}

\newtheorem{definition}{Definition}

\newcommand{\adots}{\reflectbox{$\ddots$}}
\newcommand{\ncfrac}[2]{\genfrac{}{}{0pt}{0}{#1}{#2}}
\ifCLASSINFOpdf

\else

\fi

\begin{document}

\title{Robustness Enhancement of Consensus Networks: the Optimal Memory Depth}

\author{Jiamin~Wang,~Jian~Liu,~Feng~Xiao,~Haibin~Duan,~Yuanshi~Zheng
\thanks{This work is funded by the National Natural Science Foundation of China under Grant 62273267, the Key R\&D program of Shaanxi Province under Grant 2025CY-YBXM-100, and the China Postdoctoral Science Foundation under Grants GZB20250433 and 2025M781658.
The material in this paper was not presented at any conference.
\textit{(Corresponding author: Yuanshi Zheng)}}
\thanks{
	J. Wang, J. Liu, and Y. Zheng are with the Center for Multi-Agent,
School of Mechano-electronic Engineering, Xidian University, Xi'an,
710071, China.
E-mails: wangjiamin21@outlook.com; 
liujianzym@outlook.com; yiqunzhang@xidian.edu.cn; zhengyuanshi2005@163.com.
}
\thanks{
	F. Xiao is with the School of Control and Computer Engineering, North China Electric Power University, Beijing, 102206, China (E-mail: fengxiao@ncepu.edu.cn).
}
\thanks{
	H. Duan is with the School of Automation Science and Electrical Engineering, Beihang University, Beijing, 100083, China (E-mail: hbduan@buaa.edu.cn).
}
}
\maketitle

\begin{abstract}
Understanding what governs collective robustness and how it can be enhanced remains a central pursuit in network science.
This paper investigates the robustness of multi-agent consensus networks, quantified by the $H_2$ performance metric,
and delves into the enhancing effect of agents' local memory on it.
Inspired by the hierarchical temporal structure of memory observed in neuroscience,
we focus on the role of memory depth,
which reflects the temporal features of memory from recent to remote.
Building on linear extrapolation,
we propose a consensus protocol with single-step memory and tunable memory depth,
derive the necessary and sufficient condition for achieving consensus,
and show that the protocol exhibits an inheritable consensus property across memory depths.
Furthermore, 
analytical expressions for the $H_2$ performance metric, which depend on the memory factor, memory depth, coupling gain, and Laplacian spectrum, are established.
Under balanced usage of real-time and memory information,
we demonstrate that memory at any accessible depth enhances $H_2$ performance,
and the optimal memory depth occurs at either the most recent or the most remote memory, contingent upon certain parameter regions.
Further detailed discussions are provided to clarify the broader implications of our findings.
\end{abstract}


\IEEEpeerreviewmaketitle

\section{Introduction\label{Sec:1}}
\IEEEPARstart{I}{nspired} 
by the ubiquitous collective behaviors in natural and social swarms,
where sophisticated global functionalities stem from simple local interactions,
multi-agent networks with diverse individual characteristics and interaction patterns \cite{YuanshiZheng2011,JunyiYang2020,LiqiZhou2022,Jiamin2025SCL} have been designed to accomplish various coordinated control tasks.
However, interconnectivity among agents is a double-edged sword,
since it also allows local disturbances to propagate catastrophically,
as evidenced in complex networks \cite{OriolArtime2024}.
Consequently, it is paramount to disentangle the underlying principles that govern the robustness of multi-agent networks against disturbances,
thereby providing systematic pathways for enhancing robustness.

The study of robustness in multi-agent networks usually commences with the consensus network,
as it forms the cornerstone of many other coordinated control problems.
In general, the robustness of multi-agent consensus networks to vanishing bounded disturbances can be characterized by the $H_\infty$ norm of the closed-loop transfer function matrix, which represents the worst-case level of disturbance attenuation.
In contrast, the robustness to non-vanishing noises is typically captured by the square $H_2$ norm, which quantifies the quality of noise propagation throughout the network.
These two measures are commonly referred to as the $H_\infty$ and $H_2$ performance metric, respectively.
Most existing studies employ robust control methods in conjunction with the linear matrix inequality (LMI) technique to design and optimize distributed robust controllers \cite{ZhongkuiLi2011,ShimingChen2020,HaoZhang2023,ZhongkuiLi2023,HeLi2025}, ensuring that multi-agent networks achieve prescribed $H_\infty$ and $H_2$ performance levels.

In a quite different research avenue, some studies have revealed intrinsic relationships between robustness performance indices and network topological properties, with a primary focus on continuous-time multi-agent consensus networks.
By identifying the analytical relationship, it can be straightforwardly and systematically applied to guide the design of consensus networks with enhanced robustness.
It is shown that the $H_\infty$ performance is usually determined by a single Laplacian eigenvalue, whereas the $H_2$ performance depends on all Laplacian eigenvalues.
Specifically, for consensus networks with first-order dynamics over undirected graphs, the $H_\infty$ performance is inversely proportional to the minimum nonzero Laplacian eigenvalue \cite{MiladSiami2014}, while the $H_2$ performance is determined by the inverse of the harmonic mean of all Laplacian eigenvalues \cite{GeorgeForrestYoung2010}.
Interestingly, the latter result is equivalent to the effective resistance (i.e., the well-known Kirchhoff index) of the resistor network associated with the undirected graph.
These findings have subsequently been extended to directed graphs \cite{Ding2024,JiaminWang2024TAC}, leader-following multi-agent networks \cite{MohammadPirani2018}, time-delay multi-agent networks \cite{ShimaDezfulian2018}, and second-order multi-agent networks \cite{H.GirayOral2017,YixiongFeng2018,JiaminWang2024AUTO}.
Beyond the topological perspective, it should be noted that the impact of sampling control, which alters dynamic behaviors, on the $H_\infty$ performance was investigated in \cite{JiaminWang2024TAC}.
Furthermore, the impact of protocol structure on the $H_\infty$ performance of second-order consensus networks was analyzed \cite{JiaminWang2024AUTO},
in which a concise graph criterion for protocol selection was proposed.

Despite the extensive research on continuous-time consensus networks,
insights into the intrinsic determinants of robustness in discrete-time consensus networks remain remarkably scarce.
Moreover, 
the inherent structural limitation of conventional consensus algorithms leaves little space for developing novel performance enhancement methodologies.
In other words, conventional consensus protocols suffer from a kind of \textit{amnesia}, emphasizing instantaneous yet disposable information exchange among agents, where each agent's previous state is forgotten once its local information is updated.
Even though networks may inherently involve memory effects,
such as data shifts resulting from communication delays,
they are often treated as minor technical details in the algorithmic formulation rather than being deliberately exploited (See \cite{FengXiao2016,YafengLi2024,ZhenhuaWang2025}).
This treatment overlooks potential benefits that the agent's memory of past information might offer in consensus networks.
However, recent studies have highlighted some advantages of incorporating memory information into consensus protocols,
including enhanced privacy preservation \cite{GuilhermeRamos2023} and accelerated convergence rate \cite{TuncerCanAysal2009,Jing-WenYi2023,JiahaoDai2024,GianniPasolini2020,JiahaoDai2022}.
Specifically, Ramos et al. \cite{GuilhermeRamos2023} proposed a parameter-free discrete-time consensus algorithm to obtain higher privacy index,
which needs previous states of local neighbors.
To accelerate the convergence rate of consensus algorithms,
an early study \cite{TuncerCanAysal2009} introduced a linear predictor-based approach,
in which each agent’s local update relies on a convex sum of the conventional consensus term and a multi-step memory aggregation of its own past states.
It was shown that this design achieves a faster convergence rate than the standard consensus algorithm using the optimal weighting matrix.
Building upon this foundation, Yi et al. \cite{Jing-WenYi2023} further generalized the memory-based consensus algorithm and derived analytical expressions for the optimal convergence rate and control parameters under different memory window lengths.
Recognizing that the existing results mainly concentrated on first-order multi-agent networks,
Dai et al. \cite{JiahaoDai2024} proposed a second-order consensus protocol with velocity memory,
and proved that the velocity memory from immediately previous time step can accelerate convergence compared with memoryless schemes.
However, in all the aforementioned schemes, each agent was only able to access and exploit its own memory information, without utilizing the memory of neighbors.
To address this limitation, Pasolini et al. \cite{GianniPasolini2020} proposed a novel consensus protocol that relies on agents’ historical relative measurements of their neighbors.
Dai et al. \cite{JiahaoDai2022} provided a unified form of existing memory-based consensus protocols,
incorporating memory information from both the agents themselves and their neighbors.
The demonstrated effectiveness of memory mechanisms in consensus performances has also motivated advancements in related fields, such as accelerated distributed optimization \cite{XiaoxingRen2022}, accelerated Nash equilibrium seeking \cite{JiuanGao2025}, and minimum-time consensus prediction \cite{FulongHu2025}.
In our latest work \cite{JiaminWang2025TAC},
we proved that invoking the memory information of agents and their neighbors from immediately previous time step enhances the $H_\infty$ performance of first-order discrete-time consensus networks. 

In summary, although the beneficial effect of memory information on accelerating convergence has been widely acknowledged,
its role in enhancing robustness remains underexplored.
Furthermore, findings from neuroscience indicate that short-term and long-term memories shape qualitatively different cognitive responses.
By analogy, agent's memories recalled from different temporal distances may lead to distinct collective performances in consensus networks,
yet a systematic understanding of such effects remains elusive.
Motivated by neuroscience evidence that memory formation and retrieval exhibit hierarchical temporal structures \cite{TakashiKitamura2017},
the concept of \textit{memory depth} is introduced to characterize the temporal extent to which an agent invoke past information.
A larger memory depth corresponds to retrieve more remote memory.
This naturally raises a fundamental question: \textit{how does memory depth influence the robustness of consensus networks?}

With all aforementioned observations,
this paper focuses on the $H_2$ performance of discrete-time consensus networks,
and we aim to investigate the effect of memory depth on it. 
Our contributions are summarized as follows:
\begin{itemize}
	\item[1)] A concise consensus protocol with single-step memory and tunable memory depth is proposed.
	It is a fundamental form generalized from related works such that our conclusions are potential to be extended to a wide range of existing architectures.
	Furthermore, we provide the necessary and sufficient condition for achieving consensus,
	and show that this protocol exhibits an inheritable consensus property with respect to the memory depth,
	meaning that consensus attained under remote memory ensures consensus under recent memory.
	\item[2)] We derive an analytic expression of the $H_2$ performance under the proposed protocol,
	which depends on the memory factor, memory depth, coupling gain, and Laplacian spectrum.
	Notably, for the typical case where real-time and memory information are used in balance, we establish a more explicit form of this expression, characterized by a series of continued fractions.
	\item[3)] It is proved that invoking single-step memory information of any depth that ensures consensus always improves the $H_2$ performance of consensus networks under a balanced utilization of real-time and memory information.
	Moreover, we identify the optimal memory depth within certain parameter regions,
	thereby providing effective criteria for memory depth selection between  long-term memory and short-term memory.
\end{itemize}

The rest of this paper is organized as follows.
In Section \ref{Sec:2}, some preliminaries are presented.
The problem formulation is provided in Section \ref{Sec:3}.
Consensus analysis and robustness analysis are given in Section \ref{Sec:4}.
Further implications of our results are discussed in Section \ref{Sec:5}.
Finally, Section \ref{Sec:6} concludes this paper.

\textbf{Notations}:
Throughout this paper,
$\mathbb{N}$ refers to the set of non-negative integers,
$\mathbb{N}^+$ denotes the set of positive integers,
$\mathbb{R}$ is the set of real numbers,
$\mathbb{R}^n$ represents the
$n$-dimensional real column vector space, 
and $\mathbb{R}^{n\times m}$ denotes the $n\times m$ real matrix 
space.
For $x\in\mathbb{R}$,
$\lfloor x \rfloor$ denotes the floor function and gives the largest integer less than or equal to $x$.
The all-zero and all-one matrices of appropriate dimensions are denoted by $\mathbf{0}$ and $\mathbf{1}$, respectively.
$\mathbf{0}_n$ and $\mathbf{1}_n$ refer to the $n \times 1$ column vectors of all zeros and all ones, respectively.
$I_n$ indicates the $n$-dimensional identity matrix.
$\mathrm{diag}\{a_1, \dots, a_n\}$ designates a diagonal matrix whose $i$-th diagonal element is $a_i$.
For a matrix $M$, its transpose transpose and conjugate transpose are denoted by $M^\top$ and $M^*$, respectively.
$\det(\cdot)$ and $\mathrm{Tr}(\cdot)$ respectively represent the determinant and the trace of a square matrix.
The imaginary unit is denoted by $\mathbf{j}$.
$|\cdot|$ refers to the absolute value of a real number or the cardinality of a set.
For integers $a,b\in\mathbb{N}$ with $a<b$,
let $\mathcal{I}_{a,b} = \{a,a+1,\dots,b\}$ be the set of all integers from $a$ to $b$ inclusive.
The empty set is represented as $\emptyset$,
and the expectation operator is denoted by $\mathbb{E}[\cdot]$.
\section{Preliminaries}\label{Sec:2}
Consider a group of agents with a communication network represented by a weighted undirected graph $\mathcal{G}=(\mathcal{V},\mathcal{E},\mathcal{A})$, where $\mathcal{V}=\{v_1,\dots,v_n\}$ is the vertex set, $\mathcal{E}\subseteq \mathcal{V}\times\mathcal{V}$ is the edge set, and $\mathcal{A}=[a_{ij}]\in\mathbb{R}^{n\times n}$ is the symmetric adjacency matrix with $a_{ij}\geq 0$ and $a_{ii}=0$.
Edge $\varepsilon_{ij}=(v_i,v_j)\in\mathcal{E}$ implies that $v_i$ and $v_j$ can exchange information, with weight $a_{ij}>0$; otherwise $a_{ij}=0$. 
A \textit{path} between distinct vertices $v_i$ and $v_j$ is a sequence of distinct edges $(v_i,v_{k_1}),\dots,(v_{k_l},v_j)$. $\mathcal{G}$ is \textit{connected} if a path exists between any pair of distinct vertices. 
The Laplacian matrix $L=[l_{ij}]\in\mathbb{R}^{n\times n}$ associated with $\mathcal{G}$ is defined as $l_{ii}=\sum_{j=1,j\neq i}^n a_{ij}$, $l_{ij}=-a_{ij}$ for $i\neq j$.
The neighbor set of $v_i$ is $\mathcal{N}_i=\{v_j\in\mathcal{V}:(v_i,v_j)\in\mathcal{E}\}$.
If $\mathcal{G}$ is connected, the eigenvalues of $L$ can be sorted as $0=\lambda_1<\lambda_2\leq\cdots\leq\lambda_n$.

The following definitions and lemmas will be used to derive the main results.
\begin{lemma}[\cite{YangLiu2012}]\label{lem:bar_Psi_L}
	Consider a symmetric matrix $\Psi=I_n-\frac{1}{n}\mathbf{1}_n\mathbf{1}_n^\top$ and a connected undirected graph $\mathcal{G}$ with a Laplacian matrix $L$.
	Then, there exists an orthogonal matrix $U\in\mathbb{R}^{n\times n}$ (i.e., $U^\top U=UU^\top=I_n$) such that
	\begin{equation}\label{eq:bar_Psi}
		U^\top\Psi U
		=
		\bar{\Psi}
		=
		\begin{bmatrix}
			0 & \mathbf{0}_{n-1}^{\top} \\
			\mathbf{0}_{n-1} & I_{n-1}
		\end{bmatrix}
	\end{equation}
	and
	\begin{equation}\label{eq:bar_L}
		U^\top LU
		=
		\bar{L}
		=
		\begin{bmatrix}
			0 & \mathbf{0}_{n-1}^{\top} \\
			\mathbf{0}_{n-1} & \tilde{L}
		\end{bmatrix},
	\end{equation}
	where $\tilde{L}\in\mathbb{R}^{(n-1)\times(n-1)}$ is positive definite.
\end{lemma}
\begin{definition}[Continued fractions, \cite{WilliamJones1980}]\label{def:continued_fraction}
	Given two sequences of real numbers $r_1,r_2,\dots$ and $s_0,s_1,\dots$, a fraction of the form
	\begin{equation}\label{eq:continued_fraction}
		s_0 + \cfrac{r_1}{s_1 + \cfrac{r_2}{s_2 + \ncfrac{\vphantom{0}}{\ddots}}}
	\end{equation}
	is called a \textit{continued fraction} and denoted as $s_0+\mathcal{K}_{i=1}^\infty(\frac{r_i}{s_i})$,
	while
	\begin{equation*}
		s_0+\mathcal{K}_{i=1}^n(\frac{r_i}{s_i})
		=s_0 + \cfrac{r_1}{s_1 + \cfrac{r_2}{\ddots + \cfrac{r_n}{s_n}}}
	\end{equation*}
	is termed as the $n$-th approximant of \eqref{eq:continued_fraction}.
\end{definition}
\begin{lemma}\label{lem:cfrac_property}
	Let $\mathscr{F}_n(\tau)=s_0+\mathcal{K}_{i=1}^n(\frac{r_i}{s_i})$ and $\mathscr{G}_n(\tau)=y_0+\mathcal{K}_{i=1}^n(\frac{v_i}{y_i})$,
	where
	\begin{equation*}
		\begin{cases}
			r_1=\dots=r_n=-1, \\
			s_0=\dots=s_{n-1}=\frac{2}{\tau}, s_n=1, \\
			v_1=\dots=v_{n-1}=-1, v_n=-\tau, \mathrm{if}~n>1,\\
 			v_n=-\tau, \mathrm{if}~n=1,\\
			y_0=\dots=y_{n-1}=\frac{2}{\tau}, y_n=1,\\
			n\in\mathbb{N}^+, \tau\in(-1,0)\cup(0,1).
		\end{cases}
	\end{equation*}
	Then, for all $n\in\mathbb{N}^+$, we have
	\begin{equation}\label{eq:Gn}
		\begin{cases}
			\mathscr{G}_{n+1}(\tau)>\mathscr{G}_{n}(\tau)>1, & \mathrm{if}~\tau\in(0,1), \\
			\mathscr{G}_{n+1}(\tau)<\mathscr{G}_{n}(\tau)<-1, & \mathrm{if}~\tau\in(-1,0),
		\end{cases}
	\end{equation}
	and
	\begin{equation}\label{eq:Fn}
		\begin{cases}
			\mathscr{F}_{n+1}(\tau)>\mathscr{G}_{n}(\tau)>\mathscr{F}_{n}(\tau)>1, & \mathrm{if}~\tau\in(0,1), \\
			\mathscr{F}_{n}(\tau)<\mathscr{F}_{n+1}(\tau)<\mathscr{G}_{n}(\tau)<-1, & \mathrm{if}~\tau\in(-1,0).
		\end{cases}
	\end{equation}
\end{lemma}
\begin{proof}
	The proof is provided in the Appendix \ref{app:pro_cfrac}.
\end{proof}
For notational convenience in subsequent analysis, let $\mathscr{F}_0(\tau)=1$ and $\mathscr{G}_0(\tau)=\frac{1}{\tau}$.

Consider the following discrete-time system:
\begin{equation}\label{eq:example_DT_system}
	\begin{cases}
		x(t+1)=Ax(t)+Bu(t)+D\omega(t) \\
		y(t)=Cx(t)
	\end{cases}
\end{equation}
where $x(t)\in\mathbb{R}^{n}$ is the state, $u(t)\in\mathbb{R}^{p}$ is control input in the form of $u(t)=Kx(t)$, $y(t)\in\mathbb{R}^{q}$ is the output, $\omega(t)\in\mathbb{R}^{m}$ is the exogenous input signal, $A$, $B$, $C$, and $D$ are constant matrices with appropriate dimensions.
The transfer matrix from $\omega(t)$ to $y(t)$ can be given by $T(z)=C(zI_n-\bar{A})^{-1}D$, where $\bar{A}=A+BK$.

\begin{definition}[\cite{AlexeyBelov2018}]\label{def:H2_norm}
	If system \eqref{eq:example_DT_system} is asymptotically stable,
	the $H_2$ norm of $T (z)$ is given by
	\begin{equation*}
		\|T (z)\|_2
		=\sqrt{\frac{1}{2\pi}\int_{0}^{2\pi}\mathrm{Tr}\big[T(e^{\mathbf{j}\upsilon})^{^{\scriptstyle*}}T(e^{\mathbf{j}\upsilon})\big]\mathrm{d}\upsilon}.
	\end{equation*}
\end{definition}
\begin{lemma}[\cite{Chi-TsongChen1999}]\label{lem:unique_Gramian}
	If the pair $(\bar{A},B)$ is controllable and system \eqref{eq:example_DT_system} is asymptotically stable, then the unique solution of
	\begin{equation}\label{eq:Lyapunov_equation_example}
		\bar{A}\mathcal{W}\bar{A}^\top+BB^\top=\mathcal{W}
	\end{equation}
	is positive definite.
	This solution is called the discrete \textit{controllability Gramian} and can be expressed as $\mathcal{W}=\sum_{k=0}^{\infty}A^kBB^\top (A^\top)^k$.
\end{lemma}
\begin{lemma}[\cite{KeminZhou1996}]\label{lem:H2_trace}
	The $H_2$ norm of $T(z)$ can be computed as
	\begin{equation*}
		\|T (z)\|_2=\sqrt{\mathrm{Tr}(C\mathcal{W}C^\top)},
	\end{equation*}
	where $\mathcal{W}$ is the controllability Gramian in \eqref{eq:Lyapunov_equation_example}.
\end{lemma}
\section{Problem statement}\label{Sec:3}
Consider a multi-agent network composed of $n$ agents.
The communication relationship among these agents is modeled by a connected weighted undirected graph $\mathcal{G}$,
and each agent takes the following discrete-time dynamics
\begin{equation}\label{eq:DT_dynamics}
	x_i(t+1)=u_i(t)+\omega_i(t), i\in\mathcal{I}_{1,n},
\end{equation}
where $x_i(t)\in\mathbb{R}$ is the state, $u_i(t)\in\mathbb{R}$ is the control input, and $\omega_i(t)\in\mathbb{R}$ denote the white noise with zero mean and unitary covariance.
Here, we design a memory-based consensus protocol with tunable memory depth 
\begin{equation}\label{eq:memory_protocol}
	u_i(t)=\alpha\varphi_i(t)+(1-\alpha)\varphi_i(t-\theta)
\end{equation}
where
\begin{equation*}
	\varphi_i(t)=x_i(t)+\beta\sum_{j\in\mathcal{N}_i}a_{ij}[x_j(t)-x_i(t)],
\end{equation*}
$\alpha\in(0,1)$ is termed the memory factor, $\theta\in\mathcal{I}_{1,\hat{\theta}}$ is called the memory depth, indicating how long ago the invoked memory information was, $\hat{\theta}\in\mathbb{N}^+$ denotes the maximum accessible memory depth and characterizes the agent's limited memory capacity, $\varphi_i(t-\theta)$ is referred to as the $\theta$-depth memory, $\beta$ is the positive coupling gain, and $a_{ij}$ is the adjacency element.
The multi-agent network described by \eqref{eq:DT_dynamics} and \eqref{eq:memory_protocol} is said to achieve consensus if $\lim_{t\to\infty}[x_i(t)-x_j(t)]=0$ $\forall i,j\in\mathcal{I}_{1,n}$ under any initial conditions $x_i(-\theta),\dots,x_i(0)$.

Define the consensus error $\epsilon_i(t)=x_i(t)-\frac{1}{n}\sum_{j=1}^{n}x_j(t)$ to capture the disagreement among agents induced by the noises.
By substituting \eqref{eq:memory_protocol} into \eqref{eq:DT_dynamics} and letting the collective consensus error $\epsilon=[\epsilon_1(t),\dots,\epsilon_n(t)]^\top$ serve as the nominal output, we obtain the following augmented closed-loop system:
\begin{equation}\label{eq:augmented_system}
	\begin{cases}
		\bar{x}(t+1)
		=
		\Theta\bar{x}(t)
		+
		\begin{bmatrix}
			\mathbf{0}\\
			I_n
		\end{bmatrix}
		\omega(t),  \\
		\epsilon(t)
		=
		\begin{bmatrix}
			\mathbf{0} & \Psi
		\end{bmatrix}
		\bar{x}(t),
	\end{cases}
\end{equation}
where $\bar{x}(t)=[x(t-\theta)^\top,\dots,x(t)^\top]^\top\in\mathbb{R}^{h}$ is the aggregated state vector, $h=n(\theta+1)$, $x(t)=[x_1(t),\dots,x_n(t)]^\top$, 
$\omega(t)=[\omega_1(t),\dots,\omega_n(t)]^\top$,
and
$\Theta\in\mathbb{R}^{h\times h}$ is in the form of
\begin{equation}
	\Theta=\begin{bmatrix}
		\mathbf{0} &  I_n & \cdots & \mathbf{0} \\
		\vdots & \vdots & \ddots & \vdots \\
		\mathbf{0} & \mathbf{0} & \cdots &  I_n \\
		(1-\alpha)(I_n-\beta L) & \mathbf{0} & \cdots & \alpha(I_n-\beta L)
	\end{bmatrix}.
\end{equation}

Then, the robustness of the multi-agent consensus network can be given as the following steady-state (total) mean-square deviation
\begin{equation}\label{eq:time_domain_metric}
	\rho=\lim\limits_{t\to\infty}\mathbb{E}\big[\epsilon(t)^\top\epsilon(t)\big].
\end{equation}
A smaller value of this metric indicates better robustness against noises.
It is worth mentioning that since $\rho=\|T_{\alpha,\beta,\theta}(z)\|_2^2$,
this metric is also called the $H_2$ performance metric.
Hereafter, we analyze the robustness of consensus networks using $\|T_{\alpha,\beta,\theta}(z)\|_2^2$, which is more computationally tractable than the time-domain analysis for \eqref{eq:time_domain_metric}.

In this paper, by establishing the analytic expression of the $H_2$ performance metric, we will investigate how memory depth influences the robustness of consensus networks.
Specifically, we try to examine whether the invocation of memory information leads to robustness potentiation,
and compare the robustness under long-term memory (i.e., larger $\theta$) and short-term memory (i.e., smaller $\theta$) to determine the optimal memory depth.

\section{Main results}\label{Sec:4}
In this section, we will first provide consensus analysis (see Section \ref{subsec:consensus_analysis}) for the memory-based consensus protocol,
which is necessary for subsequent robustness analysis.
Then, in Section \ref{subsec:robustness_analysis},
we give a general expression of the $H_2$ performance metric and more explicit forms for a representative memory mechanism.
Based on these results, we reveal the effect of memory depth on the robustness of consensus networks.
\subsection{Consensus Analysis}\label{subsec:consensus_analysis}
For a connected undirected graph $\mathcal{G}$, the associated Laplacian matrix $L$ is symmetric.
According to the spectral theorem \cite{RogerHorn2012}, there exists an orthogonal matrix $V\in\mathbb{R}^{n\times n}$ such that $V^\top LV=\Lambda$, where $\Lambda=\mathrm{diag}\{\lambda_1,\dots,\lambda_n\}$ comprise all eigenvalues of $L$. 
Thus, for system \eqref{eq:augmented_system}, the characteristic polynomial of $\Theta$ can be given by
\begin{align*}
	&\det(\gamma I_{h}-\Theta) \\
	&=\left|\gamma^{\theta+1}I_n-\gamma^\theta\alpha(I_n-\beta \Lambda)-(1-\alpha)(I_n-\beta \Lambda)\right| \\
	&=\prod_{i=1}^{n}\mathscr{P}(\theta,\gamma,\lambda_{i}),
\end{align*}
where $\mathscr{P}(\theta,\gamma,\lambda_{i})=\gamma^{\theta+1}-\gamma^{\theta}\alpha(1-\beta\lambda_i)-(1-\alpha)(1-\beta\lambda_i)$.
Let $\gamma_{i,k}$ ($k=1,\dots,\theta+1$) be the eigenvalues of $\Theta$ associated with $\lambda_{i}$.
Here, we provide a necessary and sufficient condition for the noise-free multi-agent network \eqref{eq:DT_dynamics} reaching consensus under the memory-based protocol \eqref{eq:memory_protocol}.
\begin{lemma}\label{lem:consensus_condition}
	In the absence of noises,
	the multi-agent network \eqref{eq:DT_dynamics} achieves consensus under protocol \eqref{eq:memory_protocol} if and only if $|\gamma_{i,k}|<1\ \forall i\in\mathcal{I}_{2,n}, k\in\mathcal{I}_{1,\theta+1}$.
\end{lemma}
\begin{proof}
	The proof is given in Appendix \ref{app:pro_consensus_condition}.
\end{proof}

Given a connected undirected graph $\mathcal{G}$,
the set of parameter pairs $(\alpha,\beta)$ ensuring consensus under $\theta$-depth memory,
defined by $\Omega_\theta=\{(\alpha,\beta)\in(0,1)\times(0,\infty):|\gamma_{i,k}|<1\ \forall i\in\mathcal{I}_{2,n}, k\in\mathcal{I}_{1,\theta+1}\}$,
is called the $\theta$-depth consensus region.

We further present an inheritable consensus property of the protocol \eqref{eq:memory_protocol} to facilitate subsequent robustness analysis.
\begin{theorem}\label{thm:consensus_property}
	For fixed memory factor $\alpha$ and coupling gain $\beta\in(0,\frac{2}{\lambda_n})$,
	if the multi-agent network \eqref{eq:DT_dynamics} achieves consensus under protocol \eqref{eq:memory_protocol} with $(\theta+1)$-depth memory in the absence of noises,
	then the consensus is still guaranteed with $\theta$-depth memory.
\end{theorem}
\begin{proof}
	The proof can be found in Appendix \ref{app:pro_consensus_property}.
\end{proof}

By analogy, Theorem \ref{thm:consensus_property} is valid for any memory depth.
Consequently, if the multi-agent network \eqref{eq:DT_dynamics} achieves consensus under protocol \eqref{eq:memory_protocol} with $\theta$-depth memory in the absence of noises,
then the consensus is also guaranteed for any memory depth smaller than $\theta$.
Furthermore, the condition $\beta\in(0,\frac{2}{\lambda_n})$ ensures consensus for the memoryless case (i.e., $\alpha=1$),
as established by the classic first-order discrete-time consensus problem \cite{RezaOlfati-Saber2007}.
These results form the basis for the subsequent robustness comparison between memoryless, short-term memory, and long-term memory cases,
since the $H_2$ performance metric is well-defined only when the consensus, in the absence of noises, is guaranteed.

\subsection{Robustness Analysis}\label{subsec:robustness_analysis}
In this subsection, based on the property $\rho=\|T_{\alpha,\beta,\theta}(z)\|_2^2$,
we investigate the robustness of the multi-agent consensus network \eqref{eq:DT_dynamics} by analyzing $\|T_{\alpha,\beta,\theta}(z)\|_2^2$,
and derive its analytic expression.
Then, we investigate how the memory depth influences the robustness.

Definition \ref{def:H2_norm} shows that the $H_2$ norm is well-defined only for asymptotically stable systems.
However, the augmented system \eqref{eq:augmented_system} contains marginally stable mode.
Prior to any further analysis, we must examine the existence of $\|T_{\alpha,\beta,\theta}(z)\|_2$.

As $\mathcal{G}$ is connected,
it follows from Lemma \ref{lem:bar_Psi_L} that there exists an orthogonal matrix $U\in\mathbb{R}^{n\times n}$ such that \eqref{eq:bar_Psi} and \eqref{eq:bar_L} hold.
Note that the matrix $\tilde{L}\in\mathbb{R}^{(n-1)\times(n-1)}$ in \eqref{eq:bar_L} is positive definite and contains all nonzero eigenvalues of $L$.
Thus, there is an orthogonal matrix
$\tilde{V}\in\mathbb{R}^{(n-1)\times(n-1)}$
such that 
$\tilde{V}^\top\tilde{L}\tilde{V}=\tilde{\Lambda}=\mathrm{diag}\{\lambda_2,\dots,\lambda_n\}$.
Let
$Q=U\bar{V}$,
\begin{equation*}
		\bar{V}=
	\begin{bmatrix}
		1 & \mathrm{0}_{n-1}^\top \\
		\mathrm{0}_{n-1} & \tilde{V}
	\end{bmatrix},~
	\bar{Q}=\begin{bmatrix}
		Q & & \\
		& \ddots & \\
		& & Q
	\end{bmatrix}\in\mathbb{R}^{h\times h}.
\end{equation*}
Performing the following orthogonal transformation
\begin{equation*}
	\tilde{x}(t)=\bar{Q}^\top\bar{x}(t),
	\tilde{\omega}(t)=Q^\top\omega(t),
	\tilde{\epsilon}(t)=Q^\top\epsilon(t)
\end{equation*}
for system \eqref{eq:augmented_system} yields
\begin{equation}\label{eq:augmented_system_2}
	\begin{cases}
		\tilde{x}(t+1)=
		\underbrace{\begin{bmatrix}
			\mathbf{0} &  I_n & \cdots & \mathbf{0} \\
			\vdots & \vdots & \ddots & \vdots \\
			\mathbf{0} & \mathbf{0} & \cdots &  I_n \\
			(1-\alpha)\bar{\Phi} & \mathbf{0} & \cdots & \alpha\bar{\Phi}
		\end{bmatrix}}_{\tilde{\Theta}}
		\tilde{x}(t)
		+
		\begin{bmatrix}
			\mathbf{0}\\
			I_n
		\end{bmatrix}
		\tilde{\omega}(t), \\
		\tilde{\epsilon}(t)=
		\begin{bmatrix}
			\mathbf{0} & \bar{\Psi}
		\end{bmatrix}\tilde{x}(t),
	\end{cases}
\end{equation}
where $\tilde{\Theta}\in\mathbb{R}^{h\times h}$,
$\bar{\Phi}=
\begin{bmatrix}
	1 & 0 \\
	0 & \Phi
\end{bmatrix}$,
and $\Phi=I_{n-1}-\beta\tilde{\Lambda}$.
Decomposing \eqref{eq:augmented_system_2} gives the following subsystem 
\begin{equation}\label{eq:augmented_system_3}
	\begin{cases}
		\hat{x}(t+1)=
		\!\underbrace{\begin{bmatrix}
			\mathbf{0} &  \!\!\!\!\! I_{n-1} & \!\!\cdots & \!\!\!\!\mathbf{0} \\
			\vdots & \!\!\!\!\!\vdots & \!\!\ddots & \!\!\!\!\vdots \\
			\mathbf{0} & \!\!\!\!\!\mathbf{0} & \!\!\cdots &  \!\!\!\! I_{n-1} \\
			(1-\alpha)\Phi & \!\!\!\!\!\mathbf{0} & \!\!\cdots & \!\!\!\!\alpha\Phi
		\end{bmatrix}}_{\hat{\Theta}}
		\hat{x}(t)
		+
		\mathcal{B}
		\hat{\omega}(t), \\
		\hat{\epsilon}(t)=
		\mathcal{C}\hat{x}(t),
	\end{cases}
\end{equation}
where $\hat{\Theta}\in\mathbb{R}^{\hat{h}\times \hat{h}}$,
$\mathcal{C}=\mathcal{B}^\top=\begin{bmatrix}\mathbf{0} & I_{n-1}\end{bmatrix}\in\mathbb{R}^{(n-1)\times\hat{h}}$,
and $\hat{h}=(n-1)(\theta+1)$.
Let $T_{1}(z)$ and $T_{2}(z)$ be the transfer matrices of systems \eqref{eq:augmented_system_2} and \eqref{eq:augmented_system_3}, respectively.
Then, the following lemma proves the existence of $\|T_{\alpha,\beta,\theta}(z)\|_2$ and provides its equivalent formulations.

\begin{lemma}\label{lem:H2_property}
	For a memory depth $\theta$,
	if $(\alpha,\beta)\in\Omega_\theta$,
	then $\|T_{\alpha,\beta,\theta}(z)\|_2=\|T_{1}(z)\|_2=\|T_{2}(z)\|_2$.
\end{lemma}
\begin{proof}
	The proof is presented in Appendix \ref{app:pro_H2_property}.
\end{proof}

\begin{table*}[htbp]
	\centering
	\caption{$\Gamma_i$ and $\xi_i$ under different memory depth}
	\label{tab:II}
	\begin{talltblr}
		[
		label = none,
		note{$\star$} = {$\psi_i=\frac{\alpha\zeta_i\phi_i}{\eta_i}-(1-\alpha)\phi_i.
		~^\dagger\chi_{i,1}
		=
		\begin{cases}
			\frac{1-2\alpha}{1-\alpha}\phi_i-1, & \mathrm{if}~\theta~\mathrm{is~odd}, \\
			\frac{1-2\alpha}{1-\alpha}\phi_i^2-1, & \mathrm{if}~\theta~\mathrm{is~even}.
		\end{cases}
		~^\ddagger\chi_{i,2}
		=
		\begin{cases}
			\frac{\alpha}{1-\alpha}, & \mathrm{if}~\theta~\mathrm{is~odd}, \\
			\frac{\alpha\phi_i}{1-\alpha}, & \mathrm{if}~\theta~\mathrm{is~even}.
		\end{cases}$}
		]
		{
			colspec={c|c|c}
		}
		\hline[1pt]
		 $\theta$ &  $\Gamma_i$ & $\xi_i$ \\
		\hline
		$\theta=1$ & $-\psi_i-1$\TblrNote{$\star$} & $\frac{\alpha\phi_i}{\eta_i}$ \\
		\hline
		$\theta=2$ & $-\psi_i\phi_i-1$ & $\frac{\alpha\phi_i^2}{\eta_i}$ \\
		\hline
		$\theta=3,4$
		& $\begin{bmatrix}
			\alpha\phi_i &  -(1-\alpha)\psi_i\phi_i-1 \\
			\chi_{i,1}\TblrNote{$\dagger$} & \chi_{i,2}\TblrNote{$\ddagger$}  \\
		\end{bmatrix}$
		& $\begin{bmatrix}
			\frac{\alpha(1-\alpha)\phi_i^2}{\eta_i} \\ 0
		\end{bmatrix}$ \\
		\hline
		$\theta>4$
		& $\begin{bmatrix}
			\mathbf{0} &  &  & \alpha\phi_i & -(1-\alpha)\psi_i\phi_i-1 \\
			&  & \alpha\phi_i & (1-2\alpha)\phi_i^2-1 & \alpha\phi_i \\
			& \adots & \adots & \adots &  \\
			\alpha\phi_i & (1-2\alpha)\phi_i^2-1 & \alpha\phi_i &  &  \\
			\chi_{i,1} & \chi_{i,2} &  &  & \mathbf{0} 
		\end{bmatrix}$
		& $\begin{bmatrix}
			\frac{\alpha(1-\alpha)\phi_i^2}{\eta_i} \\ \mathbf{0}
		\end{bmatrix}$ \\  
		\hline[1pt]
	\end{talltblr}
\end{table*}

Based on this property,
we propose a universal approach for building the analytic expression of $\|T_{\alpha,\beta,\theta}\|_2^2$ under any memory depth.
\begin{theorem}\label{thm:H2_calculation}
	For a memory depth $\theta$ and $\beta\in(0,\frac{2}{\lambda_{n}})$,
	if $(\alpha,\beta)\in\Omega_\theta$,
	we have
	\begin{equation}\label{eq:H2_analytic_expression}
		\|T_{\alpha,\beta,\theta}(z)\|_2^2
		=n-1-|\mathcal{S}|+\sum_{i\in\mathcal{S}}\frac{-1-\zeta_i w_i}{\eta_i},
	\end{equation}
	where $\zeta_i=2\alpha(1-\alpha)\phi_i^2$,
	$\eta_i=[(1-\alpha)^2+\alpha^2]\phi_i^2-1$,
	$\phi_i=1-\beta\lambda_{i}$,
	and $\mathcal{S}=\{i\in\mathcal{I}_{2,n}:\phi_i\neq 0\}$.
	In addition, $w_i$ is the last component of the unique solution to
	\begin{equation}\label{eq:linear_equations}
		\Gamma_i w=\xi_i,
	\end{equation}
	where $\Gamma_i\in\mathbb{R}^{\iota\times\iota}$ and $\xi_i\in\mathbb{R}^{\iota}$ are specified in Table \ref{tab:II},
	$w\in\mathbb{R}^\iota$,
	and $\iota=\lfloor \frac{\theta+1}{2} \rfloor$.
\end{theorem}
\begin{proof}
	The controllability matrix of system \eqref{eq:augmented_system_3} is denoted by
	\begin{equation*}
		\begin{bmatrix}
			\mathcal{B} & \hat{\Theta}\mathcal{B} & \hat{\Theta}^2\mathcal{B} & \cdots & \hat{\Theta}^{\hat{h}-1}\mathcal{B}
		\end{bmatrix}.
	\end{equation*}
	It can be verified that it has full row rank,
	which means that system \eqref{eq:augmented_system_3} is controllable.
	Furthermore, as $\mathcal{G}$ is connected and $(\alpha,\beta)\in\Omega_\theta$,
	system \eqref{eq:augmented_system_3} is asymptotically stable (see the proof of Lemma \ref{lem:H2_property}).
	It follows from Lemma \ref{lem:unique_Gramian}  that there is a unique and positive definite solution of
	\begin{equation}\label{eq:Lyapunov_equation}
		\hat{\Theta}\mathcal{W}\hat{\Theta}^\top+\mathcal{B}\mathcal{B}^\top=\mathcal{W},
	\end{equation}
	which is called the discrete \textit{controllability Gramian} and can be expressed as
	\begin{equation}\label{eq:Gramian_sum}
		\mathcal{W}=\sum_{k=0}^{\infty}\hat{\Theta}^k\mathcal{B}\mathcal{B}^\top (\hat{\Theta}^\top)^k.
	\end{equation}
	Then, it infers from Lemma \ref{lem:H2_trace} that $\|T_{2}(z)\|_2^2$ is computed as
	\begin{equation}\label{eq:T2_CWC'}
		\|T_2(z)\|_2^2=\mathrm{Tr}(\mathcal{C}\mathcal{W}\mathcal{C}^\top).
	\end{equation}
	Decompose $\mathcal{W}\in\mathbb{R}^{\hat{h}\times\hat{h}}$ into the following block form:
	\begin{equation*}
		\mathcal{W}=
		\begin{bmatrix}
			W_{1,1} & \dots & W_{1,\theta+1}\\
			\vdots & \ddots & \vdots \\
			W_{\theta+1,1} & \dots & W_{\theta+1,\theta+1}
		\end{bmatrix},
	\end{equation*}
	where $W_{i,j}\in\mathbb{R}^{(n-1)\times(n-1)}$ and $W_{i,j}=W_{j,i}$ $\forall i,j\in\mathcal{I}_{1,\theta+1}$.
	It follows from \eqref{eq:Lyapunov_equation} that
	\begin{equation}\label{eq:W_Toeplitz}
		W_{1,j}=W_{2,j+1}=\cdots=W_{\theta+1-(j-1),\theta+1}~\forall j\in\mathcal{I}_{1,\theta+1},
	\end{equation}
	i.e., $\mathcal{W}$ is a block Toeplitz matrix \cite{JesusGutierrez-Gutierrez2012}.
	Thus, we have $\mathcal{C}\mathcal{W}\mathcal{C}^\top=W_{\theta+1,\theta+1}=W_{1,1}$,
	and the $H_2$ performance metric can be readily determined as long as $W_{1,1}$ is solved.

	As each block in $\hat{\Theta}$ is diagonal,
	it follows from \eqref{eq:Gramian_sum} that each block in $\mathcal{W}$ is also diagonal.
	Combining this fact with \eqref{eq:W_Toeplitz},
	we can further deduce from \eqref{eq:Lyapunov_equation} that
	\begin{equation}\label{eq:system_of_linear_equaions_1}
		(1\!-\!\alpha)\Phi W_{1,j}\!+\!\alpha\Phi W_{1,\theta+2-j}\!=\!W_{1,\theta+3-j}~\forall j\in\mathcal{I}_{2,\theta+1}
	\end{equation}
	and
	\begin{equation}\label{eq:system_of_linear_equaions_2}
			[(1-2\alpha+2\alpha^2)\Phi^2-I_{n-1}]W_{1,1}+2\alpha(1-\alpha)\Phi^2W_{1,\theta+1}+I_{n-1}=\mathbf{0},
	\end{equation}
	where $\Phi=\mathrm{diag}\{\phi_2,\dots,\phi_n\}$.
	Integrating \eqref{eq:system_of_linear_equaions_1} and \eqref{eq:system_of_linear_equaions_2} yields a matrix equation regarding $W_{1,1},\dots,W_{1,\theta+1}$,
	which is equivalent to the following small-scale system of equations
	\begin{equation}\label{eq:system_of_linear_equaions_3}
		\bar{\Gamma}_i
		\begin{bmatrix}
			\bar{w}_1^{[i]} \\ \vdots \\ \bar{w}_{\theta+1}^{[i]} 
		\end{bmatrix}
		=
		\bar{\xi}
	\end{equation}
	for all $i\in\mathcal{I}_{2,n}$,
	where $\bar{\Gamma}_i\in\mathbb{R}^{(\theta+1)\times(\theta+1)}$ is given by
	\begin{equation*}
		\bar{\Gamma}_i=
		\begin{bmatrix}
			0 & \mu_i & & & & & \nu_i & -1 \\
			& & \ddots & & & \adots & \adots & \\
			& & & \mu_i & \nu_i & -1 & & \\
			& & & \nu_i & \mu_i-1 & 0 & & \\
			& & \adots & -1 & 0 & \mu_i & & \\
			& \adots & \adots & & & & \ddots & \\
			\nu_i & -1 & & & & & & \mu_i \\
			\eta_i & 0 & & & & & & \zeta_i
		\end{bmatrix}
	\end{equation*}
	when $\theta$ is odd and
	\begin{equation*}
		\bar{\Gamma}_i=
		\begin{bmatrix}
			0 & \mu_i & & & & & & \nu_i & -1 \\
			& & \ddots & & & & \adots & \adots & \\
			& & & \mu_i & 0 & \nu_i & \adots & & \\
			& & & 0 & \phi_i & -1 & & & \\
			& & & \nu_i & -1 & \mu_i & & & \\
			& & \adots & \adots & & & \ddots & & \\
			& \adots & \adots & & & & & \ddots & \\
			\nu_i & -1 & & & & & & & \mu_i \\
			\eta_i & 0 & & & & & & & \zeta_i  
		\end{bmatrix}
	\end{equation*}
	when $\theta$ is even;
	$\nu_i=\alpha\phi_i$, $\mu_i=(1-\alpha)\phi_i$; $\bar{w}_1^{[i]},\dots,\bar{w}_{\theta+1}^{[i]}$ are the $(i-1)$-th diagonal elements of diagonal matrices $W_{1,1},\dots,W_{1,\theta+1}$ respectively;
	$\bar{\xi}=[\mathbf{0}_\theta^\top,-1]^\top$.
	As previously mentioned,
	we just need to solve $\bar{w}_{1}^{[2]},\dots,\bar{w}_{1}^{[n]}$.
	
	If $\phi_i=0$, solving \eqref{eq:system_of_linear_equaions_3} gives $\bar{w}_{1}^{[i]}=1$.
	If $\phi_i\neq0$,
    we apply a series of row operations and complete the proof by enumeration.
	
	\textbf{Case \uppercase\expandafter{\romannumeral1}}: $\theta$ is odd and $\theta>3$.

	Define the row-addition matrix as $E_{i,j}(k)=I_{\theta+1}+k\mathbf{e}_i\mathbf{e}_j^\top\in\mathbb{R}^{(\theta+1)\times(\theta+1)}$,
	where $\mathbf{e}_i$ and $\mathbf{e}_j$ are the standard basis vectors in $\mathbb{R}^{\theta+1}$.
	Let $R=\prod_{j=1}^{5}R_j$, $\hat{\Gamma}_i=R\bar{\Gamma}_i$, and $\hat{\xi}=R\bar{\xi}_i$,
	where
	\begin{align*}
		R_1&=\prod_{m=1}^{\iota-2}E_{m+1,m}(-\nu_i), \\
		R_2&=E_{\iota,\iota-1}\left(\tfrac{\alpha}{\alpha-1}\right), \\
		R_3&=\prod_{m=1}^{\iota-1}E_{\iota-m,\theta+1+m-\iota}(\mu_i), \\
		R_4&=\prod_{m=1}^{\iota-1}E_{\theta+m-\iota,\theta+1+m-\iota}(\nu_i), \\
		R_5&=E_{\theta,\theta+1}\left(-\tfrac{\nu_i}{\eta_i}\right).
	\end{align*}
	Since $(\alpha,\beta)\in\Omega_\theta$,
	it follows from Lemma \ref{lem:consensus_condition} that $\mathscr{P}(\theta,\gamma,\lambda_{i})$ is Schur stable for all $i\in\mathcal{I}_{2,n}$.
	As $0<\beta<\frac{2}{\lambda_n}$,
	there must have $-1<\phi_i<1$.
	Through simple algebra, we have $\eta_i\neq 0$ $\forall i\in\mathcal{I}_{2,n}$ and thus the row-addition matrix $R_5$ is tenable.
	
	After the above operations, $\hat{\Gamma}_i\in\mathbb{R}^{(\theta+1)\times(\theta+1)}$ can be written as the block
	\begin{equation*}
		\hat{\Gamma}_i
		=
		\begin{bmatrix}
			\mathbf{0} & \Gamma_i \\
			M_1 & M_2
		\end{bmatrix}
	\end{equation*}
	and there holds $\hat{\xi}=[\xi_i^\top,\frac{\nu_i^{\iota-1}}{\eta_i},\dots,\frac{\nu_i}{\eta_i},-1]\in\mathbb{R}^{\theta+1}$,
	where $M_1\in\mathbb{R}^{\iota\times\iota}$ is an anti-diagonal matrix,
	$M_2\in\mathbb{R}^{\iota\times\iota}$ is an upper triangular matrix,
	$\Gamma_i$ and $\xi_i$ are given in Table \ref{tab:II}.
	Apparently, $\bar{w}_{\theta+1}^{[i]}$ is subject to
	\begin{equation*}
		\begin{bmatrix}
			\mathbf{0} & \Gamma_i
		\end{bmatrix}
		\begin{bmatrix}
			\bar{w}_{1}^{[i]} \\
			\vdots \\
			\bar{w}_{\theta+1}^{[i]}
		\end{bmatrix}
		=
		\xi_i,
	\end{equation*}
	which is exactly the last component of the unique solution to \eqref{eq:linear_equations}.
	By inspecting \eqref{eq:system_of_linear_equaions_3},
	$\bar{w}_{1}^{[i]}$ can be further derived from $\eta_i\bar{w}_{1}^{[i]}+\zeta_i\bar{w}_{\theta+1}^{[i]}=-1$.
	
	In brief, we have
	\begin{equation}\label{eq:phi_0_neq0}
		\bar{w}_{1}^{[i]}=
		\begin{cases}
			1, & \mathrm{if}~\phi_i=0, \\
			\frac{-1-\zeta_i\bar{w}_{\theta+1}^{[i]}}{\eta_i}, & \mathrm{if}~\phi_i\neq0.
		\end{cases}
	\end{equation}
	Therefore, combining with the fact $\mathcal{C}\mathcal{W}\mathcal{C}^\top=W_{1,1}$,
	it follows from \eqref{eq:T2_CWC'} and Lemma \ref{lem:H2_property} that $\|T_{\alpha,\beta,\theta}(z)\|_2^2=\mathrm{Tr}(W_{1,1})=\sum_{i=2}^{n}\bar{w}_{1}^{[i]}$, that is, \eqref{eq:H2_analytic_expression} is obtained.
	
	\textbf{Case \uppercase\expandafter{\romannumeral2}}: $\theta=1$ or $\theta=3$.
		
	For $\theta=1$,
	by applying the row operations $R_5\bar{\Gamma}_i$ and $R_5\bar{\xi}_i$,
	the system of equations \eqref{eq:system_of_linear_equaions_3} becomes
	\begin{equation*}
		\begin{bmatrix}
			0 & -\psi_i-1 \\
			\eta_i & \zeta_i
		\end{bmatrix}
		\begin{bmatrix}
			\bar{w}_1^{[i]} \\ \bar{w}_{2}^{[i]} 
		\end{bmatrix}
		=
		\begin{bmatrix}
			\frac{\nu_i}{\eta_i} \\ -1
		\end{bmatrix}.
	\end{equation*}
    Since $-1<\phi_i<1$ gives $(2\alpha-1)\phi_i+1>0$ and $(\alpha-1)\phi_i-1<0$,
    it is straightforward that
	$-\psi_i-1=\frac{(\phi_i-1)[(\alpha-1)\phi_i-1][(2\alpha-1)\phi_i+1]}{\eta_i}\neq 0$.
	Therefore, $\bar{w}_{2}^{[i]}$ is solvable and \eqref{eq:phi_0_neq0} still holds.
	
	For $\theta=3$,
	the system of equations \eqref{eq:system_of_linear_equaions_3} becomes
	\begin{equation}\label{eq:theta=3_1}
	\begin{bmatrix}
		0 & \mu_i & \nu_i & -1 \\
		0 & \nu_i & \mu_i-1 & 0 \\
		\nu_i & -1 & 0 & \mu_i \\
		\eta_i & 0 & 0 & \zeta_i
	\end{bmatrix}
	\begin{bmatrix}
		\bar{w}_1^{[i]} \\ \bar{w}_{2}^{[i]} \\ \bar{w}_3^{[i]} \\ \bar{w}_{4}^{[i]}
	\end{bmatrix}
	=
	\begin{bmatrix}
		0 \\ 0 \\ 0 \\ -1
	\end{bmatrix}.
	\end{equation}
Let $\bar{R}=\prod_{j=2}^{5}R_j$.
Performing the row operations $\bar{R}\bar{\Gamma}_i$ and $\bar{R}\bar{\xi}_i$,
\eqref{eq:theta=3_1} reduces to
	\begin{equation*}
	\begin{bmatrix}
		0 & \!0 & \!\nu_i & \!-\mu_i\psi_i-1 \\
		0 & \!0 & \!\frac{1-2\alpha}{1-\alpha}\phi_i-1 & \!\frac{\alpha}{1-\alpha} \\
		0 & \!-1 & \!0 & \!-\psi_i \\
		\eta_i & \!0 & \!0 & \!\zeta_i
	\end{bmatrix}\!\!
	\begin{bmatrix}
		\bar{w}_1^{[i]} \\ \bar{w}_{2}^{[i]} \\ \bar{w}_3^{[i]} \\ \bar{w}_{4}^{[i]}
	\end{bmatrix}
	\!=\!
	\begin{bmatrix}
		\frac{\nu_i\mu_i}{\eta_i} \\ 0 \\ \frac{\nu_i}{\eta_i} \\ -1
	\end{bmatrix}.
	\end{equation*}
Therefore, $\bar{w}_{4}^{[i]}$ is obtained in terms of \eqref{eq:linear_equations},
and \eqref{eq:phi_0_neq0} still holds.

In summary, it follows from $\|T_{\alpha,\beta,\theta}(z)\|_2^2=\sum_{i=2}^{n}\bar{w}_{1}^{[i]}$ that \eqref{eq:H2_analytic_expression} is derived.   

\textbf{Case \uppercase\expandafter{\romannumeral3}}: $\theta$ is even.

For the case $\theta>4$, $\Gamma_i$ is anti-tridiagonal.
Let $\hat{R}_2=E_{\iota,\iota-1}(\frac{\nu_i}{\alpha-1})E_{\iota,\iota+1}(\phi_i)$ and $\hat{R}=R_1\hat{R}_2R_3R_4R_5$.
After the row operations $\hat{R}\bar{\Gamma}_i$ and $\hat{R}\bar{\xi}_i$,
the system of equations \eqref{eq:system_of_linear_equaions_3} can be recast as
\begin{equation*}
	\begin{bmatrix}
		\mathbf{0} & \Gamma_i \\
		M_3 & M_4
	\end{bmatrix}
	\begin{bmatrix}
		\bar{w}_{1}^{[i]} \\
		\vdots \\
		\bar{w}_{\theta+1}^{[i]}
	\end{bmatrix}
	=
	\begin{bmatrix}
		\xi_i \\
		\tilde{\xi}_i
	\end{bmatrix},
\end{equation*}
where
$\Gamma_i$ and $\xi_i$ are shown in Table \ref{tab:II},
$M_3\in\mathbb{R}^{(\iota+1)\times(\iota+1)}$ is anti-diagonal,
$M_4\in\mathbb{R}^{(\iota+1)\times\iota}$,
and $\tilde{\xi}_i=[\frac{\nu_i^{\iota}}{\eta_i},\dots,\frac{\nu_i}{\eta_i},-1]^\top$.
It is obvious that $\bar{w}_{\theta+1}^{[i]}$ can be solved via \eqref{eq:linear_equations},
and \eqref{eq:phi_0_neq0} holds.
Therefore, \eqref{eq:H2_analytic_expression} is obtained.

The proof for $\theta=2$ and $\theta=4$ can be provided in a similar vein.
For saving space, we omit it here.
\end{proof}

For smaller memory depth,
we can further derive more explicit analytic expressions of the $H_2$ performance metric.
Recall that $-\psi_i-1\neq0$ and one can similarly verify that $-\psi_i\phi_i-1\neq0$.
By solving \eqref{eq:linear_equations} and substituting the solution into \eqref{eq:H2_analytic_expression},
we can get
\begin{equation*}
	\|T_{\alpha,\beta,1}(z)\|_2^2
	=\sum_{i=2}^{n}\frac{1-\mu_i}{(1-\phi_i)(1+\mu_i)[(2\alpha-1)\phi_i+1]}
\end{equation*}
and
\begin{equation*}
	\|T_{\alpha,\beta,2}(z)\|_2^2
	=\sum_{i=2}^{n}\frac{1-\mu_i\phi_i}{(1-\phi_i^2)[1-(1-2\alpha)\mu_i\phi_i]}.
\end{equation*}
For $\theta=3,4$,
integrating the analysis of the cases $\chi_{i,1}=0$ and $\chi_{i,1}\neq0$,
we can obtain
\begin{equation*}
	\|T_{\alpha,\beta,3}(z)\|_2^2=\sum_{i=2}^{n}\frac{(1-2\alpha)\mu_i\phi_i^2-\mu_i^2-\mu_i+1}{(1-\phi_i)[(1-2\alpha)\phi_i-1][(1-2\alpha)\mu_i\phi_i^2+\mu_i^2-\mu_i-1]}
\end{equation*}
and
\begin{equation*}
	\|T_{\alpha,\beta,4}(z)\|_2^2=\sum_{i=2}^{n}\frac{(1-2\alpha)\mu_i\phi_i^3-(2-\alpha)\mu_i\phi_i+1}{(1-\phi_i^2)[(1-2\alpha)^2\mu_i\phi_i^3+(3\alpha-2)\mu_i\phi_i+1]}.
\end{equation*}

For $\theta>4$, deriving the explicit analytic expressions of the $H_2$ performance metric is just a repetitive algebraic endeavor that relies on elaborate enumeration.
Consequently, they are not provided in the present work.
However, imposing on the special configuration $\alpha=\frac{1}{2}$,
which balances the usage of real-time information and memory information,
allows for the unified derivation of all explicit analytic expressions along the memory depth.
The results are given as follows.

\begin{corollary}\label{cor:alpha=1/2}
	For a memory depth $\theta$ and $\beta\in(0,\frac{2}{\lambda_{n}})$,
	if $\alpha=\frac{1}{2}$ and $(\alpha,\beta)\in\Omega_\theta$,
	there is
	\begin{equation}\label{eq:H2_analytic_expression_1/2}
		\|T_{\frac{1}{2},\beta,\theta}(z)\|_2^2
		=\sum_{i=2}^{n}\varpi_{\theta,i}(\phi_i),
	\end{equation}
	where
	\begin{equation*}
		\varpi_{\theta,i}(\phi_i)=\frac{\phi_i-2\mathscr{H}_{\iota-1,i}}{(\phi_i^2-2)\mathscr{H}_{\iota-1,i}+\phi_i}\geq1
	\end{equation*}
	and
	\begin{equation*}
		\mathscr{H}_{\iota-1,i}=
		\begin{cases}
			\mathscr{F}_{\iota-1}(\phi_i), & \mathrm{if}~\theta~\mathrm{is~odd}, \\
			\mathscr{G}_{\iota-1}(\phi_i), & \mathrm{if}~\theta~\mathrm{is~even}.
		\end{cases}
	\end{equation*}
\end{corollary}
\begin{proof}
	For $\theta=1,2,3,4$,
	the results can be readily obtained by solving \eqref{eq:linear_equations} with $\alpha=\frac{1}{2}$ and then substituting the solution into \eqref{eq:H2_analytic_expression}.
	However, for $\theta>4$, a more sophisticated treatment is required.
	
	Define the row-scaling matrix as $D_i(k)=I_{\iota}+(k-1)\bar{\mathbf{e}}_i\bar{\mathbf{e}}_i^\top\in\mathbb{R}^{\iota\times\iota}$,
	where $\bar{\mathbf{e}}_i$ is the standard basis vector in $\mathbb{R}^{\iota}$.
	Let $\bar{E}_{i,j}(k)=I_{\iota}+k\bar{\mathbf{e}}_i\bar{\mathbf{e}}_j^\top$ be a new row-addition matrix in $\mathbb{R}^{\iota\times\iota}$,
	$\check{R}=H\bar{E}_{\iota-1,\iota}(1)D$,
	$\check{\Gamma}_i=\check{R}\Gamma_i$,
	and
	$\check{\xi}_i=\check{R}\xi_i$,
	where
	\begin{align*}
		H&=\prod_{m=1}^{\iota-2}E_{m,m+1}\left(\tfrac{1}{\mathscr{H}_{\iota-1-m,i}}\right), \\
		D&=\prod_{m=1}^{\iota-1}D_m\left(\tfrac{2}{\phi_i}\right).
	\end{align*}
	Then, the solution of \eqref{eq:linear_equations} with $\alpha=\frac{1}{2}$ is equivalent to that of $\check{\Gamma}_iw=\check{\xi}_i$,
	where
	\begin{equation*}
		\check{\Gamma}_i=
		\begin{bmatrix}
			& & & & -\mathscr{H}_{\iota-1,i}-\frac{\phi_i}{\phi_i^2-2} \\
			& & & -\mathscr{H}_{\iota-2,i} & 1 \\
			& & \adots & \adots &  \\
			& -\mathscr{H}_{1,i} & 1 &  &  \\
			-1 & 1 & &  & 
		\end{bmatrix}
	\end{equation*}
	is a lower anti-triangular matrix and $\check{\xi}_i=[\frac{\phi_i}{\phi_i^2-2},\mathbf{0}_{\iota-1}^\top]^\top$.
	
	Recall that $\beta\in(0,\frac{2}{\lambda_{n}})$ yields $-1<\phi_i<1$ $\forall i\in\mathcal{I}_{2,n}$.
	One can observe that $0<\frac{\phi_i}{\phi_i^2-2}<1$ when $\phi_i\in(-1,0)$ and $-1<\frac{\phi_i}{\phi_i^2-2}<0$ when $\phi_i\in(0,1)$.
	As evidence from Lemma \ref{lem:cfrac_property},
	there is $-\mathscr{H}_{\iota-1,i}-\frac{\phi_i}{\phi_i^2-2}\neq0$ $\forall \phi_i\in(-1,0)\cup(0,1)$.
	Then, the last component of the solution to $\check{\Gamma}_iw=\check{\xi}_i$ is $\frac{\phi_i}{(2-\phi_i^2)\mathscr{H}_{\iota-1,i}-\phi_i}$,
	thereby yielding \eqref{eq:H2_analytic_expression_1/2} according to Theorem \ref{thm:H2_calculation}.
	Furthermore, since $\varpi_{\theta,i}(\phi_i)=1$ when $\phi_i=0$,
	the presented results encompass this case as well.
\end{proof}

For simplicity, we denote $\|T_{1,\beta}(z)\|_2^2$ as the $H_2$ performance metric in the memoryless case (i.e., $\alpha=1$),
which can be given by
\begin{equation*}
	\|T_{1,\beta}(z)\|_2^2=\sum_{i=2}^{n}\frac{1}{1-\phi_i^2}.
\end{equation*}
With the explicit analytic expressions of the $H_2$ performance metric established,
we are now positioned to examine the effects of memory information.
It will undertake a comparative analysis of two aspects: the memory versus memoryless cases, and short-term versus long-term memory cases.
The results are provided as follows. 

\begin{corollary}\label{cor:opt_depth}
	Consider the memory factor $\alpha=\frac{1}{2}$ and the sequence of memory depths $1,\dots,\hat{\theta}-1,\hat{\theta}$.
	If $\beta\in(0,\frac{2}{\lambda_{n}})$ and $(\alpha,\beta)\in\Omega_{\hat{\theta}}$,
	then the following statements are deduced:
	\begin{itemize}
		\item[1)] For all $\theta\in\mathcal{I}_{1,\hat{\theta}}$,
		there holds 
		\begin{equation*}
			\|T_{\frac{1}{2},\beta,\theta}(z)\|_2^2\leq\|T_{1,\beta}(z)\|_2^2;
		\end{equation*}
		\item[2)] For all even $\theta\in\mathcal{I}_{1,\hat{\theta}-2}$,
		there is
		\begin{equation*}
			\|T_{\frac{1}{2},\beta,\theta}(z)\|_2^2
			\geq\|T_{\frac{1}{2},\beta,\theta+2}(z)\|_2^2;
		\end{equation*}
		\item[3)] For all odd $\theta\in\mathcal{I}_{1,\hat{\theta}-2}$,
		we have
		\begin{equation*}
			\|T_{\frac{1}{2},\beta,\theta}(z)\|_2^2
			\geq\|T_{\frac{1}{2},\beta,\theta+1}(z)\|_2^2
			\geq\|T_{\frac{1}{2},\beta,\theta+2}(z)\|_2^2
		\end{equation*}
		when $0<\beta\leq\frac{1}{\lambda_{n}}$,
		and we obtain
		\begin{equation*}
			\|T_{\frac{1}{2},\beta,\theta}(z)\|_2^2
			\leq\|T_{\frac{1}{2},\beta,\theta+2}(z)\|_2^2
			\leq\|T_{\frac{1}{2},\beta,\theta+1}(z)\|_2^2
		\end{equation*}
		when $\frac{1}{\lambda_{2}}\leq\beta<\frac{2}{\lambda_{n}}$;
		\item[4)] We can conclude that
		\begin{equation*}
			\min_{\theta\in\mathcal{I}_{1,\hat{\theta}}}
			\|T_{\frac{1}{2},\beta,\theta}(z)\|_2^2
			=\begin{cases}
				\|T_{\frac{1}{2},\beta,\hat{\theta}}(z)\|_2^2, & \mathrm{if}~0<\beta\leq\frac{1}{\lambda_{n}}, \\
				\|T_{\frac{1}{2},\beta,1}(z)\|_2^2, & \mathrm{if}~\frac{1}{\lambda_{2}}\leq\beta<\frac{2}{\lambda_{n}}.
			\end{cases}
		\end{equation*}
	\end{itemize}
	Moreover, the equalities in statements 1)--3) hold if and only if $\mathcal{G}$ is a complete graph with identical edge weights and $\beta=\frac{1}{\lambda_{2}}$.
\end{corollary}
\begin{proof}
	We structure the proof into four parts, addressing each statement in order.
	Recall that $\beta\in(0,\frac{2}{\lambda_{n}})$ yields $-1<\phi_i<1$ $\forall i\in\mathcal{I}_{2,n}$.
	
	\textbf{Proof of 1).} For $\beta\neq\frac{1}{\lambda_{i}}$ and $\theta\in\mathcal{I}_{1,\bar{\theta}}$,
	it follows from \eqref{eq:Gn} and \eqref{eq:Fn} that
	\begin{equation*}
		\varpi_{\theta,i}(\phi_i)(1-\phi_i^2)=1+\frac{\phi_i^2(\mathscr{H}_{\iota-1,i}-\phi_i)}{(\phi_i^2-2)\mathscr{H}_{\iota-1,i}+\phi_i}<1,
	\end{equation*}
	namely, $\varpi_{\theta,i}(\phi_i)<\frac{1}{1-\phi_i^2}$ $\forall i\in\mathcal{I}_{2,n}$.
	Once $\beta=\frac{1}{\lambda_{i}}$, we have $\varpi_{\theta,i}(\phi_i)=\frac{1}{1-\phi_i^2}$.
	Given the summation form of the $H_2$ performance metric,
	we know the statement 1) holds.
	
	\textbf{Proof of 2).} For $\beta\neq\frac{1}{\lambda_{i}}$ and even $\theta$,
	it follows from $\mathscr{F}_0(\tau)=1$, $\mathscr{G}_0(\tau)=\frac{1}{\tau}$, and \eqref{eq:Gn} that
	\begin{equation*}
		\varpi_{\theta,i}(\phi_i)-\varpi_{\theta+2,i}(\phi_i)=\frac{\phi_i^3[\mathscr{G}_{\iota}(\phi_i)-\mathscr{G}_{\iota-1}(\phi_i)]}{[(\phi_i^2-2)\mathscr{G}_{\iota-1}(\phi_i)+\phi_i][(\phi_i^2-2)\mathscr{G}_{\iota}(\phi_i)+\phi_i]}>0
	\end{equation*}
	holds for all $i\in\mathcal{I}_{2,n}$.
	As long as $\beta=\frac{1}{\lambda_{i}}$,
	it is observed that $\varpi_{\theta,i}(\phi_i)=\varpi_{\theta+2,i}(\phi_i)$ holds.
	Given the summation form of the $H_2$ performance metric,
	the statement 2) is correct.
	
	\textbf{Proof of 3).} Consider odd $\theta>1$, and the following equalities:
	\begin{align*}
		&\varpi_{\theta,i}(\phi_i)-\varpi_{\theta+1,i}(\phi_i)=\frac{\phi_i^3[\mathscr{G}_{\iota-1}(\phi_i)-\mathscr{F}_{\iota-1}(\phi_i)]}{[(\phi_i^2-2)\mathscr{F}_{\iota-1}(\phi_i)+\phi_i][(\phi_i^2-2)\mathscr{G}_{\iota-1}(\phi_i)+\phi_i]}, \\
		&\varpi_{\theta,i}(\phi_i)-\varpi_{\theta+2,i}(\phi_i)=\frac{\phi_i^3[\mathscr{F}_{\iota}(\phi_i)-\mathscr{F}_{\iota-1}(\phi_i)]}{[(\phi_i^2-2)\mathscr{F}_{\iota-1}(\phi_i)+\phi_i][(\phi_i^2-2)\mathscr{F}_{\iota}(\phi_i)+\phi_i]}, \\
		&\varpi_{\theta+1,i}(\phi_i)-\varpi_{\theta+2,i}(\phi_i)=\frac{\phi_i^3[\mathscr{F}_{\iota}(\phi_i)-\mathscr{G}_{\iota-1}(\phi_i)]}{[(\phi_i^2-2)\mathscr{G}_{\iota-1}(\phi_i)+\phi_i][(\phi_i^2-2)\mathscr{F}_{\iota}(\phi_i)+\phi_i]}.
	\end{align*}
	For $\beta\neq\frac{1}{\lambda_{i}}$,
	it is inferred from \eqref{eq:Fn} that
	\begin{equation*}
		\begin{cases}
			\varpi_{\theta,i}(\phi_i)>\varpi_{\theta+1,i}(\phi_i)>\varpi_{\theta+2,i}(\phi_i), & \mathrm{if}~0<\beta\leq\frac{1}{\lambda_{n}}, \\
			\varpi_{\theta,i}(\phi_i)<\varpi_{\theta+2,i}(\phi_i)<\varpi_{\theta+1,i}(\phi_i), & \mathrm{if}~\frac{1}{\lambda_{2}}\leq\beta<\frac{2}{\lambda_{n}}.
		\end{cases}
	\end{equation*}
	If $\beta=\frac{1}{\lambda_{i}}$,
	there holds $\varpi_{\theta,i}(\phi_i)=\varpi_{\theta+1,i}(\phi_i)=\varpi_{\theta+2,i}(\phi_i)$.
	Through simple algebra,
	one can verify that these results are still valid for $\theta=1$.
	Given the summation form of the $H_2$ performance metric,
	the statement 3) is proved.
	
	Based on the above analysis, one can observe that the equalities in statements 1)--3) hold if and only if
	$\beta=\frac{1}{\lambda_{2}}=\dots=\frac{1}{\lambda_{n}}$,
	which implies that $\mathcal{G}$ is a complete graph with identical edge weights.
	
	\textbf{Proof of 4).} This statement is an immediate consequence from statements 2) and 3).
\end{proof}

\begin{figure*}[ht!]
	\centering
	\begin{subfigure}{0.3\textwidth}
		\centering
		\includegraphics[width=\linewidth]{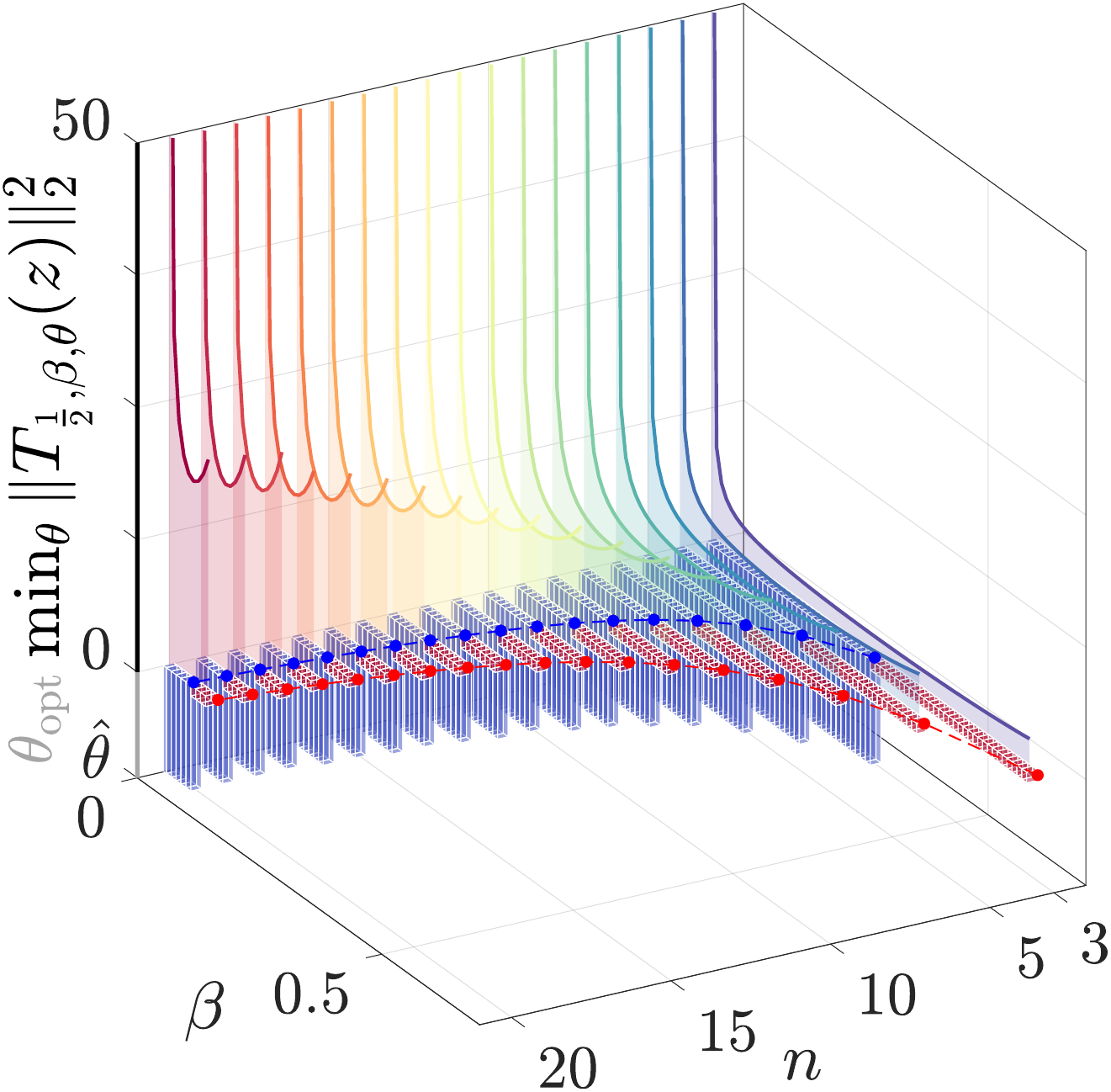}
		\caption{$K_{20}$}
		\label{Fig1-1}
	\end{subfigure}
	\begin{subfigure}{0.3\textwidth}
		\centering
		\includegraphics[width=\linewidth]{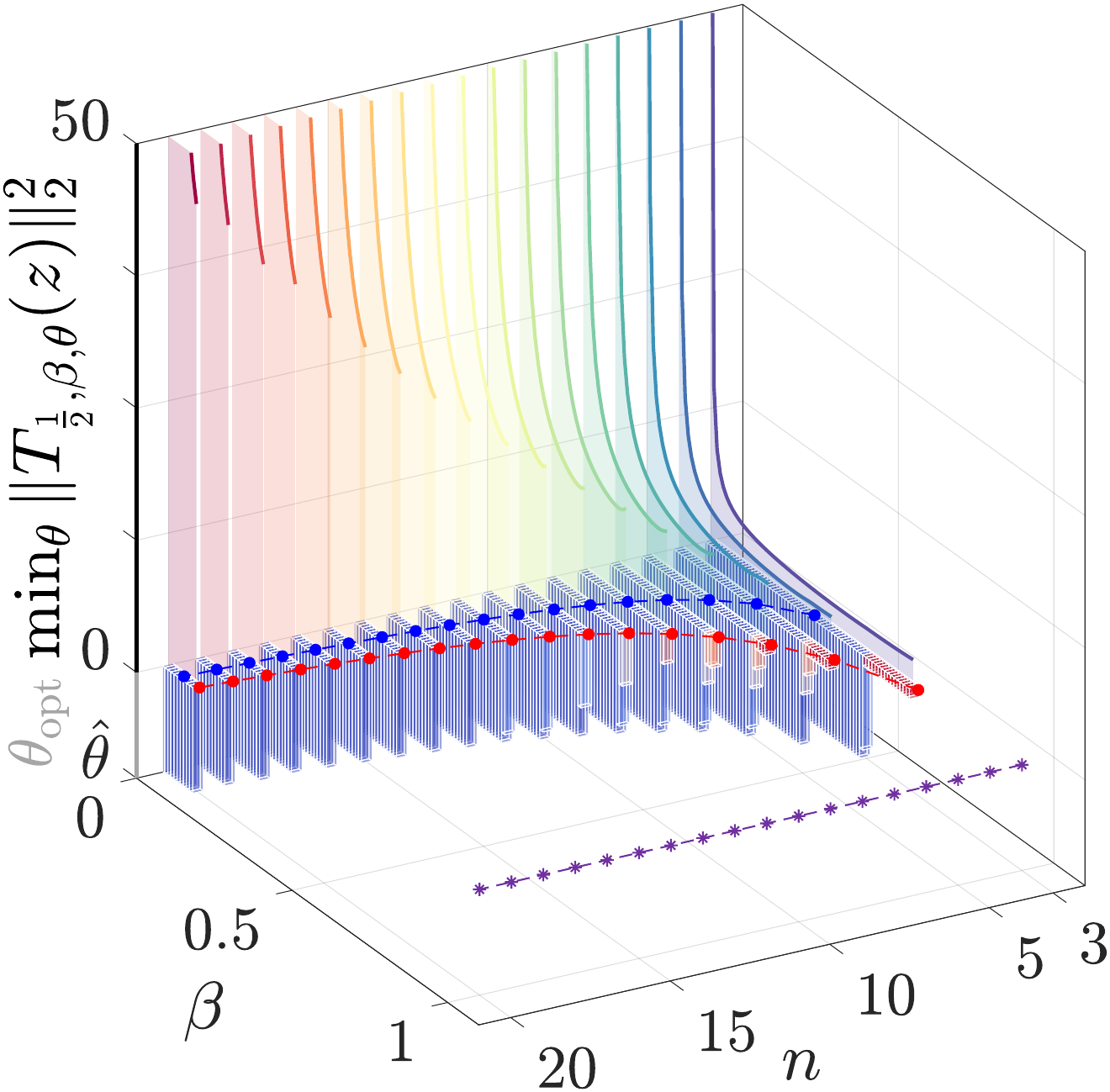}
		\caption{$S_{20}$}
		\label{Fig1-2}
	\end{subfigure}
	\begin{subfigure}{0.3\textwidth}
		\centering
		\includegraphics[width=\linewidth]{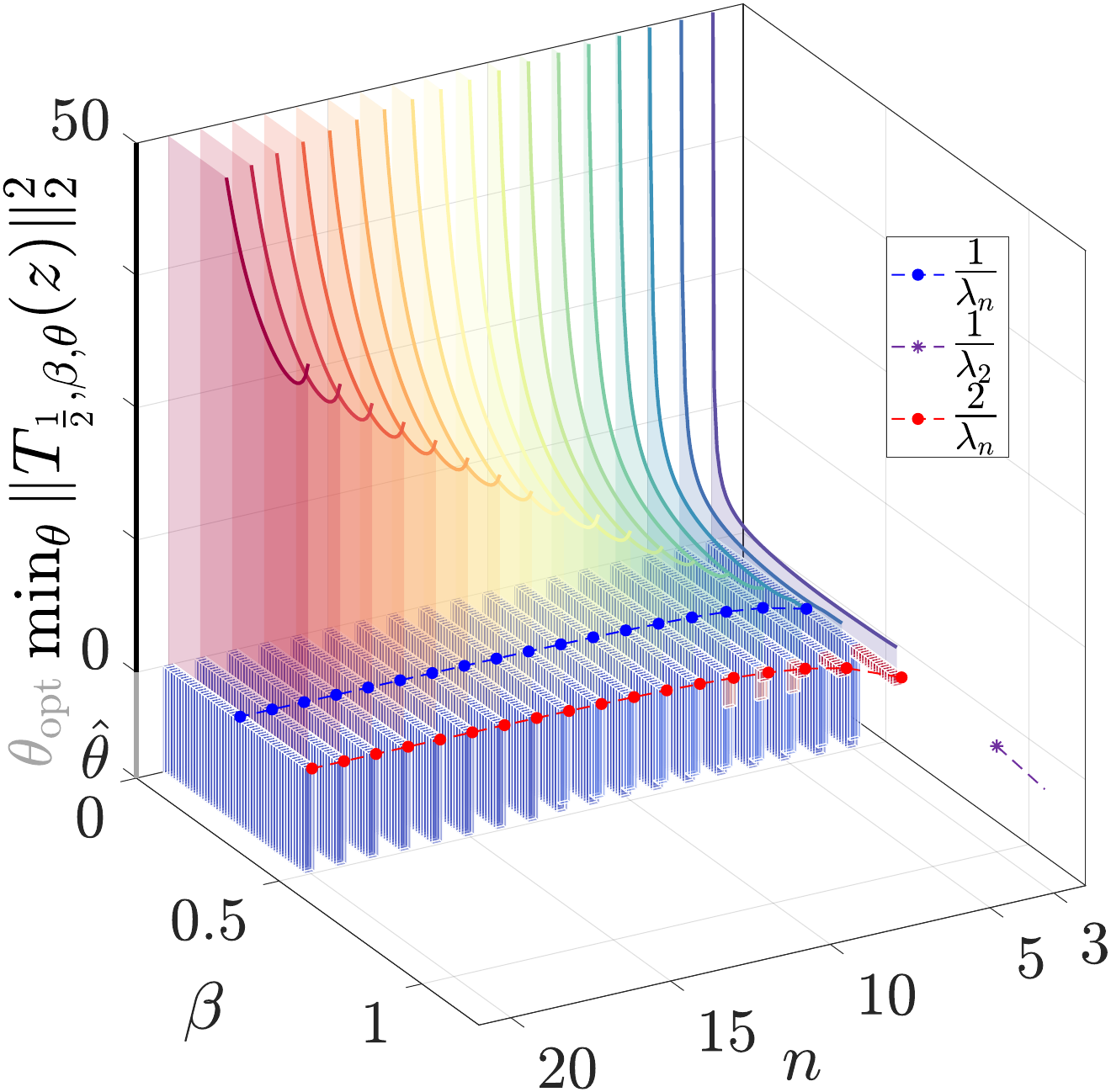}
		\caption{$P_{20}$}
		\label{Fig1-3}
	\end{subfigure}\\
	\begin{subfigure}{0.3\textwidth}
		\centering
		\includegraphics[width=\linewidth]{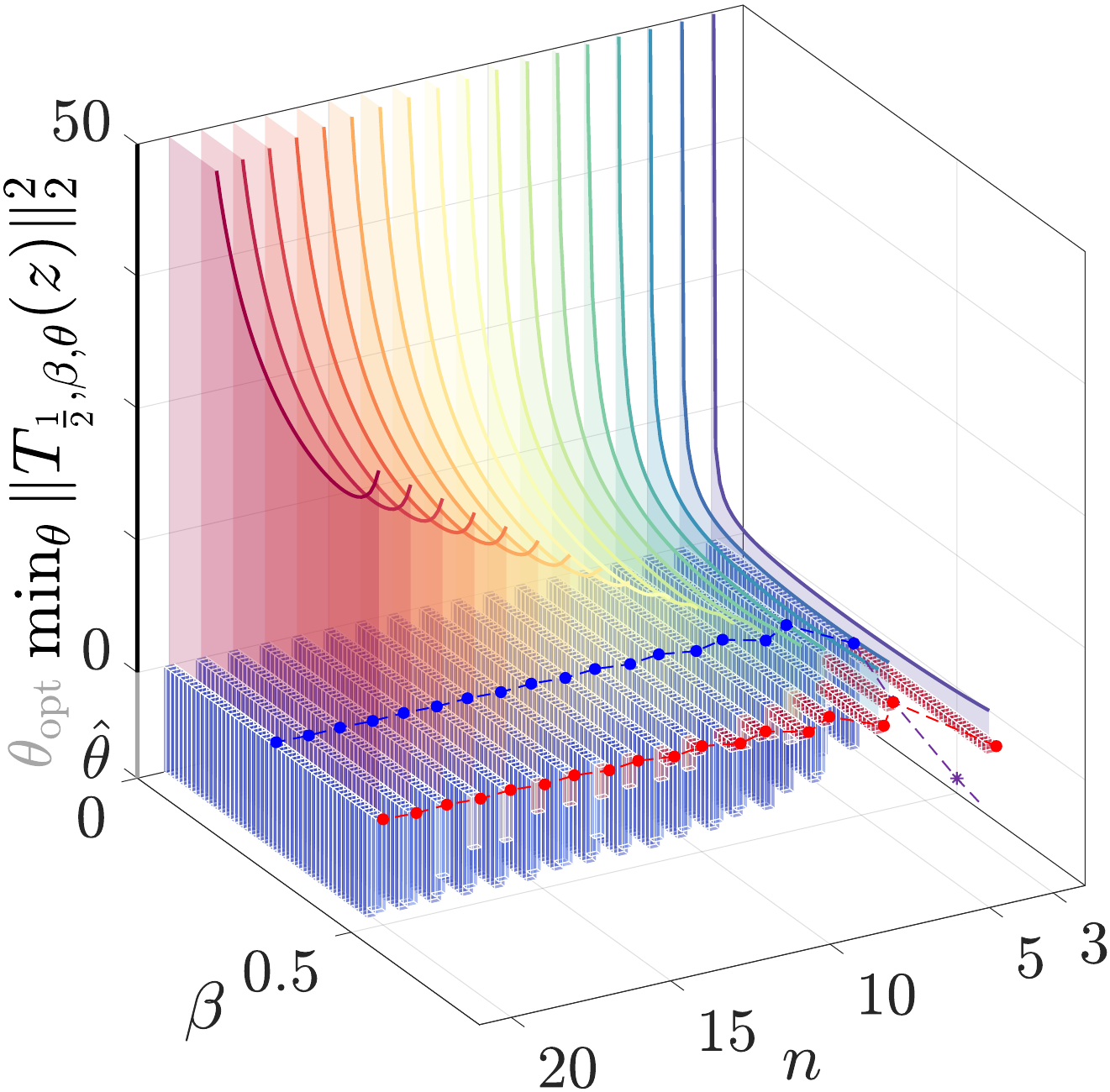}
		\caption{$C^1_{20}$}
		\label{Fig1-4}
	\end{subfigure}
	\begin{subfigure}{0.3\textwidth}
		\centering
		\includegraphics[width=\linewidth]{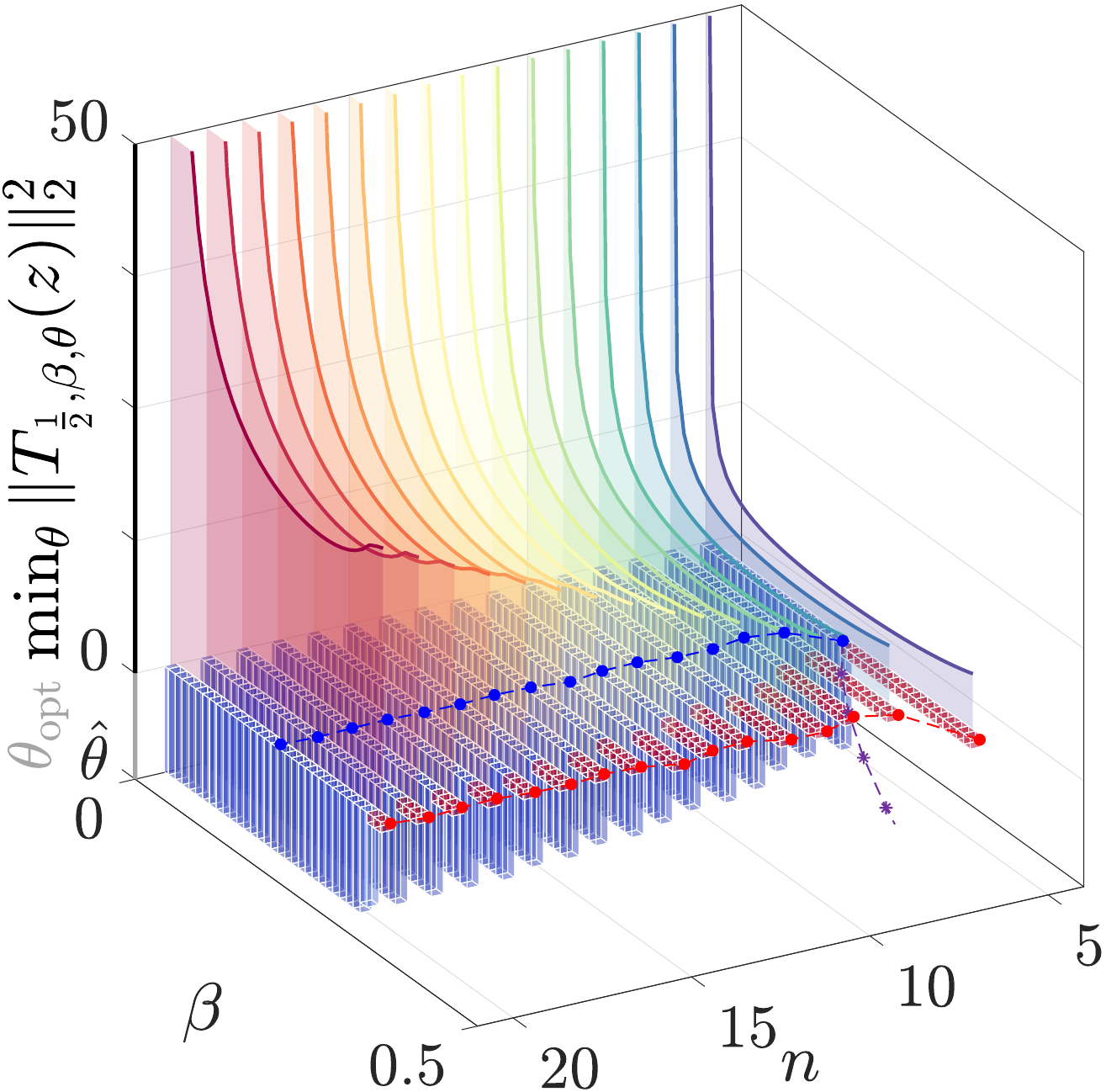}
		\caption{$C^2_{20}$}
		\label{Fig1-5}
	\end{subfigure}
	\begin{subfigure}{0.3\textwidth}
		\centering
		\includegraphics[width=\linewidth]{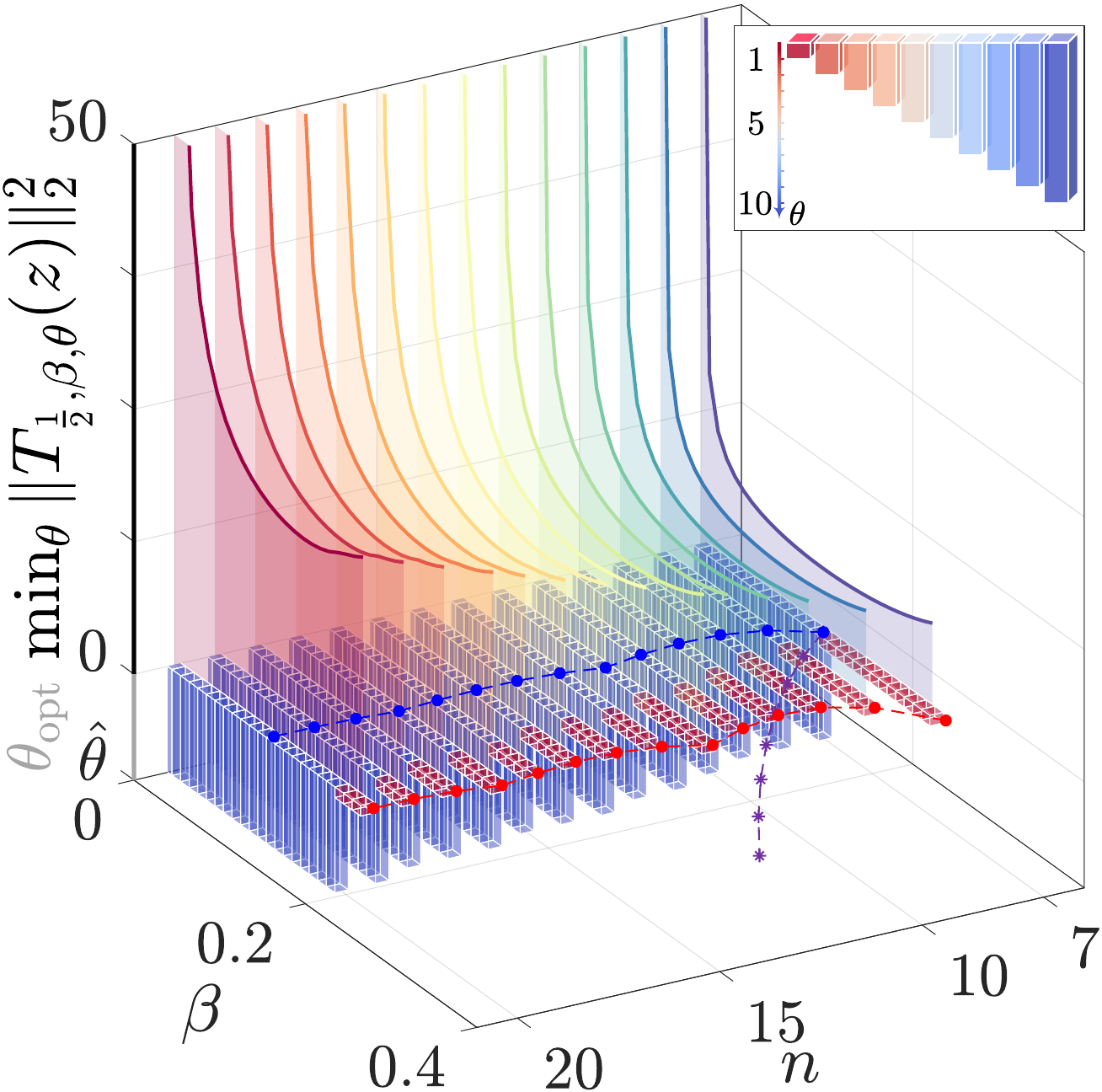}
		\caption{$C^3_{20}$}
		\label{Fig1-6}
	\end{subfigure}
	\caption{Trajectories of the optimal $H_2$ performance metrics and optimal memory depths over typical graph families.}
	\label{Fig1}
\end{figure*}

The result reveals that incorporating real-time and memory information in a balanced manner enables memory at any accessible depth to enhance the robustness of consensus networks,
although the effect is not uniform across memory depths. 
Specifically, long-term memory yields better robustness for $0<\beta\leq\frac{1}{\lambda_{n}}$,
whereas for $\frac{1}{\lambda_{2}}\leq\beta<\frac{2}{\lambda_{n}}$,
invoking only the most recent memory is optimal.
The underlying mechanism is that each component of the $H_2$ performance metric,
associated with each nonzero Laplacian eigenvalue,
is monotonically decreasing with memory depth in the former case,
but shows completely opposite behaviors for odd and even memory depths in the latter case.

In addition, statements 2) and 3) jointly reveal an intriguing phenomenon.
For large coupling gains, robustness enhances as the even memory depth increase,
whereas it deteriorates as the odd memory depth grows.
Nevertheless, the robustness at odd memory depths consistently outperforms that at the even memory depths.

Having established the optimal memory depth for the case of balanced memory invocation,
it is natural to ask what would happen if pure memory information is employed (i.e., $\alpha=0$).
Firstly, it is easy to verify that $\mathscr{P}(\theta,\gamma,\lambda_{i})=\gamma^{\theta+1}-(1-\beta\lambda_{i})$ is Schur stable for all $i\in\mathcal{I}_{2,n}$ and for all $\theta$ when $\beta\in(0,\frac{2}{\lambda_{n}})$.
Thus, Lemma \ref{lem:consensus_condition} and Theorem \ref{thm:consensus_property} are still valid.
Then, directly solving \eqref{eq:system_of_linear_equaions_3} gives the following analytic expression of $H_2$ performance metric regardless of memory depth
\begin{equation*}
		\|T_{0,\beta,\theta}(z)\|_2^2=\|T_{1,\beta}(z)\|_2^2=\sum_{i=2}^{n}\frac{1}{1-\phi_i^2},
\end{equation*}
which exhibits identical robustness to the memoryless case.
Combined with the statement 1) given in Corollary \ref{cor:opt_depth},
this result indicates that robustness can be enhanced only through a mixed utilization of memory and real-time information.

\section{Discussion}\label{Sec:5}
\subsection{Optimal Memory Depth over Typical Graph Families}
Consider the following families of undirected graphs with $n\geq3$ vertices and uniform unit edge weights:
the complete graph family $\mathbb{K}_n=\{K_3,\dots,K_n\}$, the star graph family $\mathbb{S}_n=\{S_3,\dots,S_n\}$,
the chain graph family $\mathbb{P}_n=\{P_3,\dots,P_n\}$,
and the $2d$-regular ring lattice family $\mathbb{C}^d_n=\{C_{2d+1}^d,\dots,C_{n}^d\}$,
where every $C^d_i$ (with $2d+1\leq i\leq n$) represents a network with a highly regular structure in which vertices are placed on a circle,
each connected to its $2d$ closest neighbors.
Further information on these graphs can be found in \cite{PietVanMieghem2023}.

Let $\alpha=\frac{1}{2}$, $\hat{\theta}=10$, and $n=20$.
Fig. \ref{Fig1} illustrates the trends of the optimal memory depth and the corresponding optimal $H_2$ performance over the above graph families.
Each bar’s height represents the optimal memory depth $\theta_{\mathrm{opt}}$ and is colored accordingly.
The dashed lines show the trajectories of $\frac{1}{\lambda_{2}}$, $\frac{1}{\lambda_{n}}$, and $\frac{2}{\lambda_{n}}$ with respect to the network size.
In addition, the solid curves depict the trajectories of the $H_2$ performance metric under the corresponding optimal memory depths.
We observe that, for any given network size,
there exists an optimal coupling gain that maximizes the $H_2$ performance metric.
However, the $H_2$ performance deteriorates as the network size grows,
indicating that larger-scale networks are more susceptible to noise propagation.

As we can see, for all graphs, the optimal memory depth is always equal to $1$ when $0<\beta\leq\frac{1}{\lambda_{n}}$.
The case $\frac{1}{\lambda_{2}}\leq\beta<\frac{2}{\lambda_{n}}$ occurs only in complete graphs and small-scale ring lattices,
where the optimal memory depth remains consistently $\hat{\theta}$.
These observations are consistent with the results stated in Corollary \ref{cor:opt_depth}.
For the case $\frac{1}{\lambda_{n}}\leq\beta<\frac{1}{\lambda_{2}}$,
we only need to analyze non-complete graphs.
As illustrated in Fig. \ref{Fig1},
the optimal memory depth transitions from long-term memory to short-term memory as the coupling gain increases.
This provides an important insight for optimizing the robustness of consensus networks.
Specifically, when agents rely more heavily on information from their neighbors,
employing short-term memory becomes more effective for suppressing noise.

\subsection{Effects of Unbalanced Memory Invocation}
We further examine the effect of imbalance between real-time and memory information (i.e., $\alpha\neq\frac{1}{2}$) on the three representative graphs $K_{15}$, $C^2_{15}$, and $P_{15}$,
which respectively represent the densest, a moderately connected, and the sparsest topologies.
The curves in Fig. \ref{Fig2} depict the trajectories of the $H_2$ performance metric with respect to the memory factor under different coupling gains,
with colors distinguishing the performance at various memory depths.

\begin{figure}[!ht]
	\centering
	\includegraphics[width=0.4\textwidth]{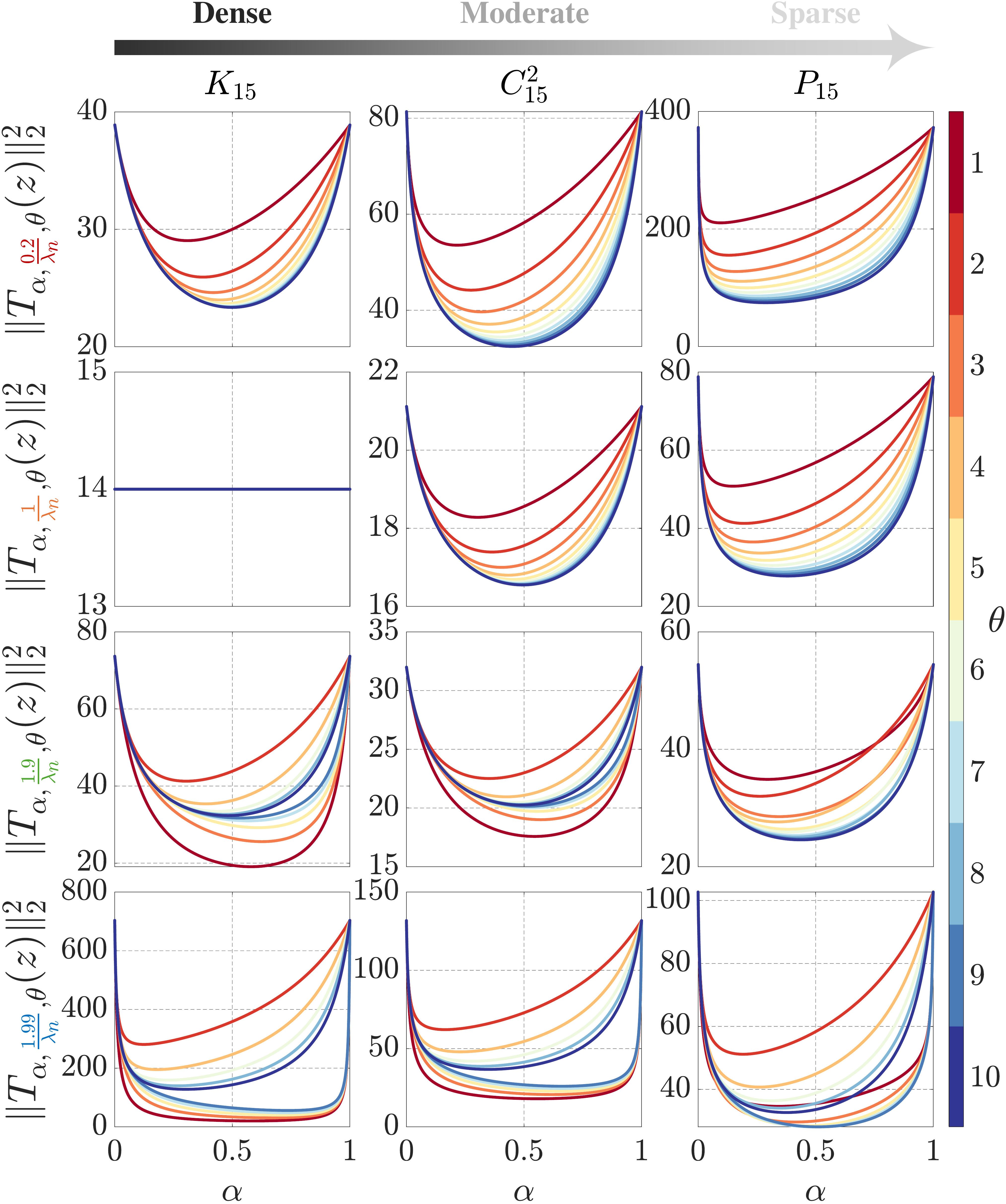}
	\caption{Trends of the $H_2$ performance metric on $K_{15}$, $C^2_{15}$, and $P_{15}$.}
	\label{Fig2}
\end{figure}

We can observe that the optimal memory depth still shifts from long-term memory to short-term memory as the coupling gain increases.
On the one hand, for small coupling gains,
the $H_2$ performance metric decreases monotonically with memory depth, regardless of the memory factor.
On the other hand, when the coupling gain becomes large, a striking pattern emerges in the denser graphs $K_{15}$ and $C^2_{15}$:
odd memory depths consistently outperform even depths in $H_2$ performance,
and the $H_2$ performance exhibits opposite monotonic trends for odd versus even memory depths.
Moreover, the hybrid strategy, combining memory and real-time information, always outperforms pure memory or memoryless strategies.
These observations confirm that statements 1)--3) in Corollary \ref{cor:opt_depth} remain fully valid for $\alpha\neq\frac{1}{2}$.
In contrast, this phenomena is absent in the sparse graph $P_{15}$.

These findings provide valuable insights for the robust design of consensus networks.
In particular, the $H_2$ performance of dense networks follows consistent and regular trends with respect to both the memory factor and memory depth,
suggesting better resilience to unbalanced memory invocation.
In contrast, sparse networks are sensitive to these parameters and require careful parameter optimization to maintain performance.

\begin{figure}[!ht]
	\centering
	\begin{subfigure}{0.35\textwidth}
		\centering
		\includegraphics[width=\linewidth]{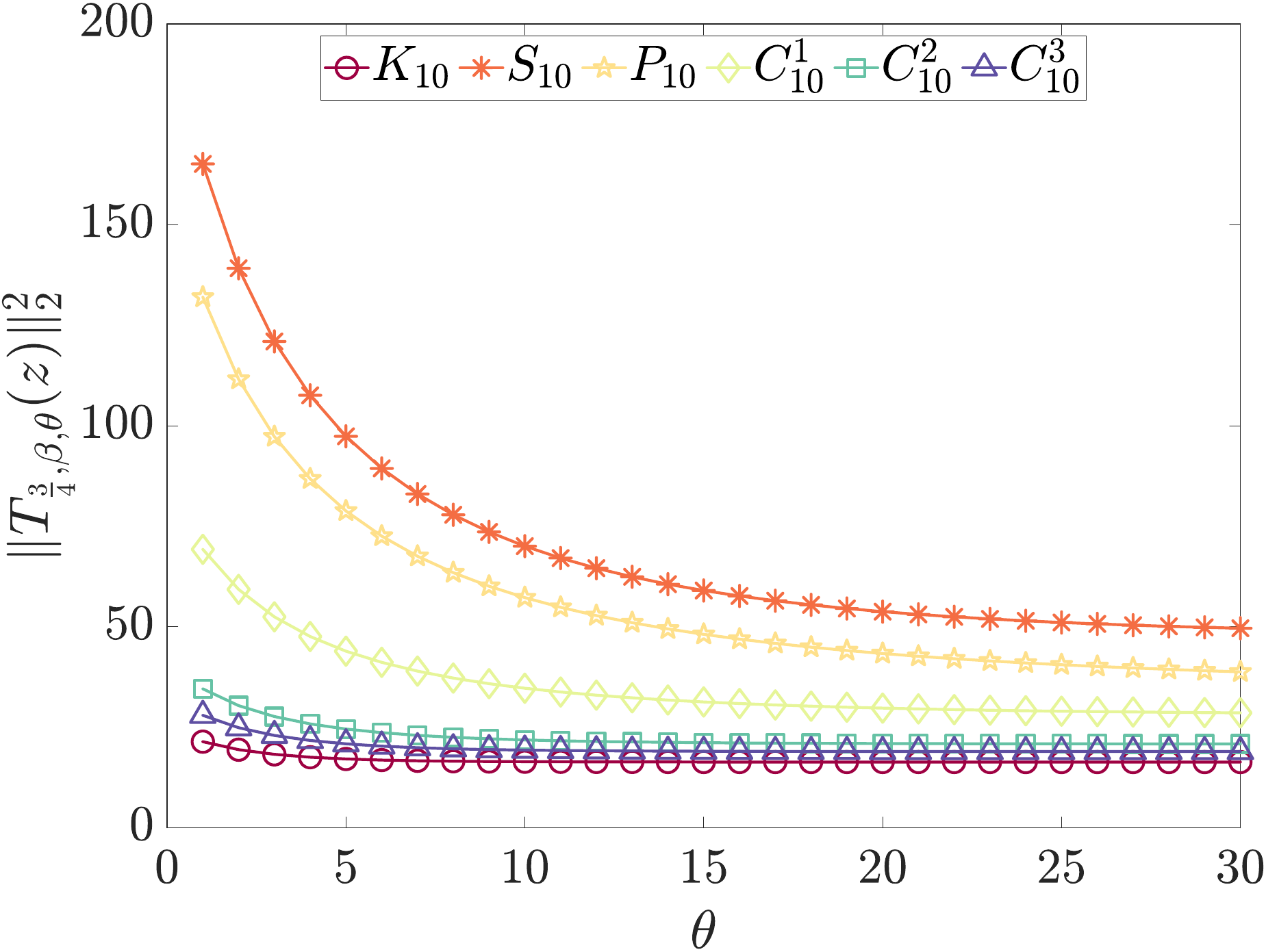}
		\caption{$\beta=\frac{0.2}{\lambda_{n}}$}
		\label{Fig3-1}
	\end{subfigure}
	\begin{subfigure}{0.35\textwidth}
		\centering
		\includegraphics[width=\linewidth]{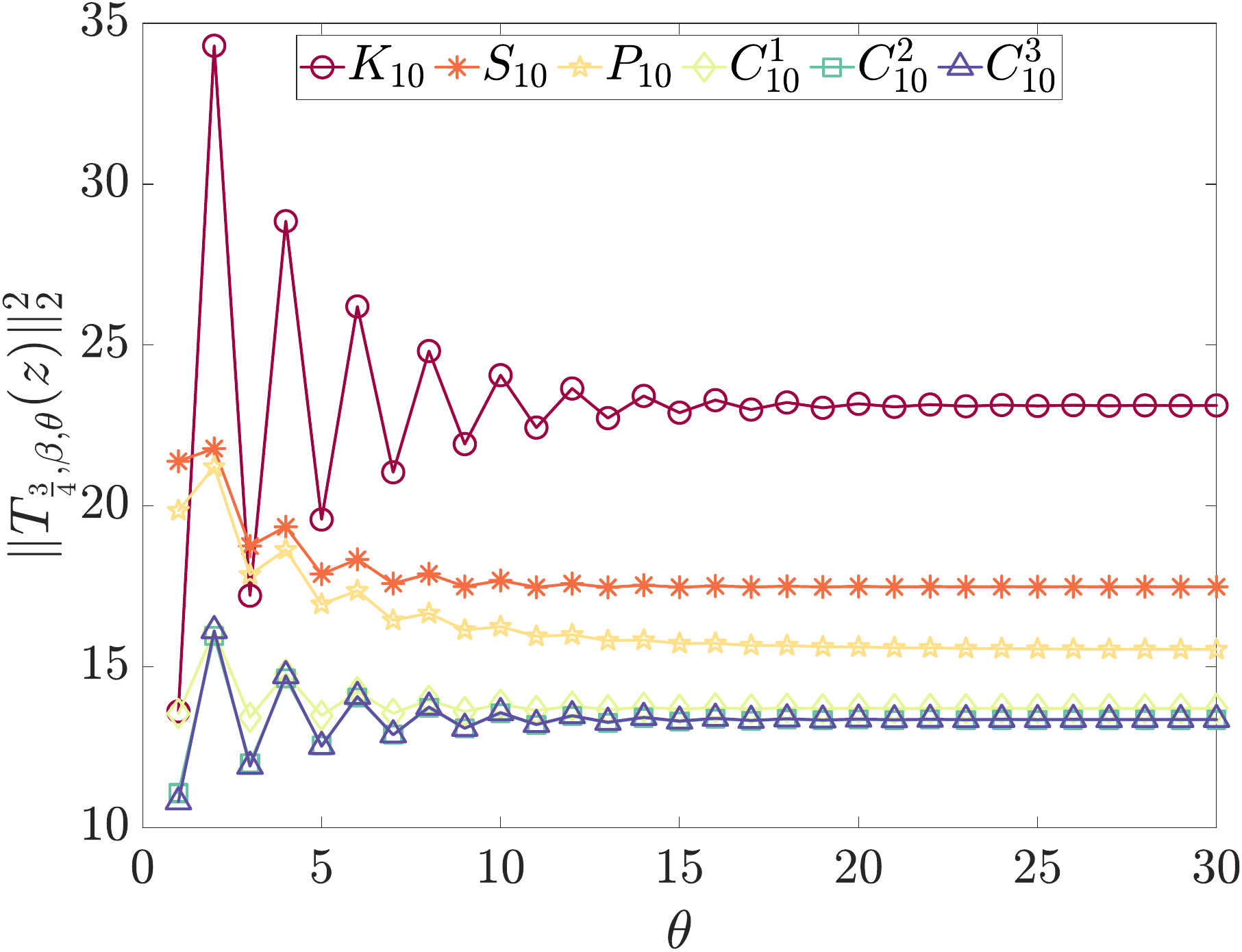}
		\caption{$\beta=\frac{1.9}{\lambda_{n}}$}
		\label{Fig3-2}
	\end{subfigure}
	\caption{Trajectories of the $H_2$ performance metric with respect to memory depth.}
	\label{Fig3}
\end{figure}

\subsection{Marginal Utility of Increasing Memory Depth}
Here, agents are considered to be omniscient,
in the sense that they can use memories of arbitrary depth.
It follows from \eqref{eq:Gn} and \eqref{eq:Fn} that
\begin{equation*}
	\lim_{\theta\to\infty}\|T_{\frac{1}{2},\beta,\theta}(z)\|_2^2
	=\sum_{i=2}^{n}\frac{2}{2-\phi_i^2},
\end{equation*}
which exhibits pronounced \textit{marginal utility} with increasing memory depth; that is, the metric converges to a constant as memory depth grows.

For $\alpha=\frac{3}{4}$,
Fig. \ref{Fig3} depicts the trajectories of the $H_2$ performance metric with respect to memory depth under two representative parameter configurations.
Specifically, as shown in Fig. \ref{Fig3-1}, access to more remote memories consistently enhances $H_2$ performance across all network topologies.
In contrast, Fig. \ref{Fig3-2} illustrates that this phenomenon is restricted to the star and chain graphs.
Furthermore, as the coupling gain varies,
the relative performance ranking of the communication topologies can reverse,
with the original optimal topology potentially becoming the worst-performing one.
Nevertheless, the \textit{marginal utility} was still observed under both configurations,
as well as in many other numerical tests not presented here.
These observations indicate that it is enough to select optimal memory depth over a finite temporal extent,
beyond which more remote memory yields negligible performance gains.
This substantially facilitates the practicality and feasibility of the robustness enhancement scheme based on the memory mechanism.

\subsection{Optimal Parameter Pair}
Another important issue is to find the optimal parameter pair $(\alpha,\beta)$ for a given communication network and memory depth.
The two-dimensional surface of $\|T_{\alpha,\beta,1}(z)\|_2^2$ with respect to $\alpha$ and $\beta$ on $C^2_{20}$ is illustrated in Fig. \ref{Fig4}.
To highlight one-dimensional optima,
we overlay two curves.
The orange solid curve represents the optimal $H_2$ performance obtained by minimizing over $\beta$,
whereas the black solid curve depicts the optimal $H_2$ performance obtained by minimizing over $\alpha$.
The global minimizer $(\alpha_{\mathrm{opt}},\beta_{\mathrm{opt}})$, marked with a red circle, lies exactly at the intersection of the two curves.
We can get similar visualizations for other graphs and memory depths.
These tests collectively indicate that the optimal parameter pair is unique.

\begin{figure}[!ht]
	\centering
	\includegraphics[width=0.4\textwidth]{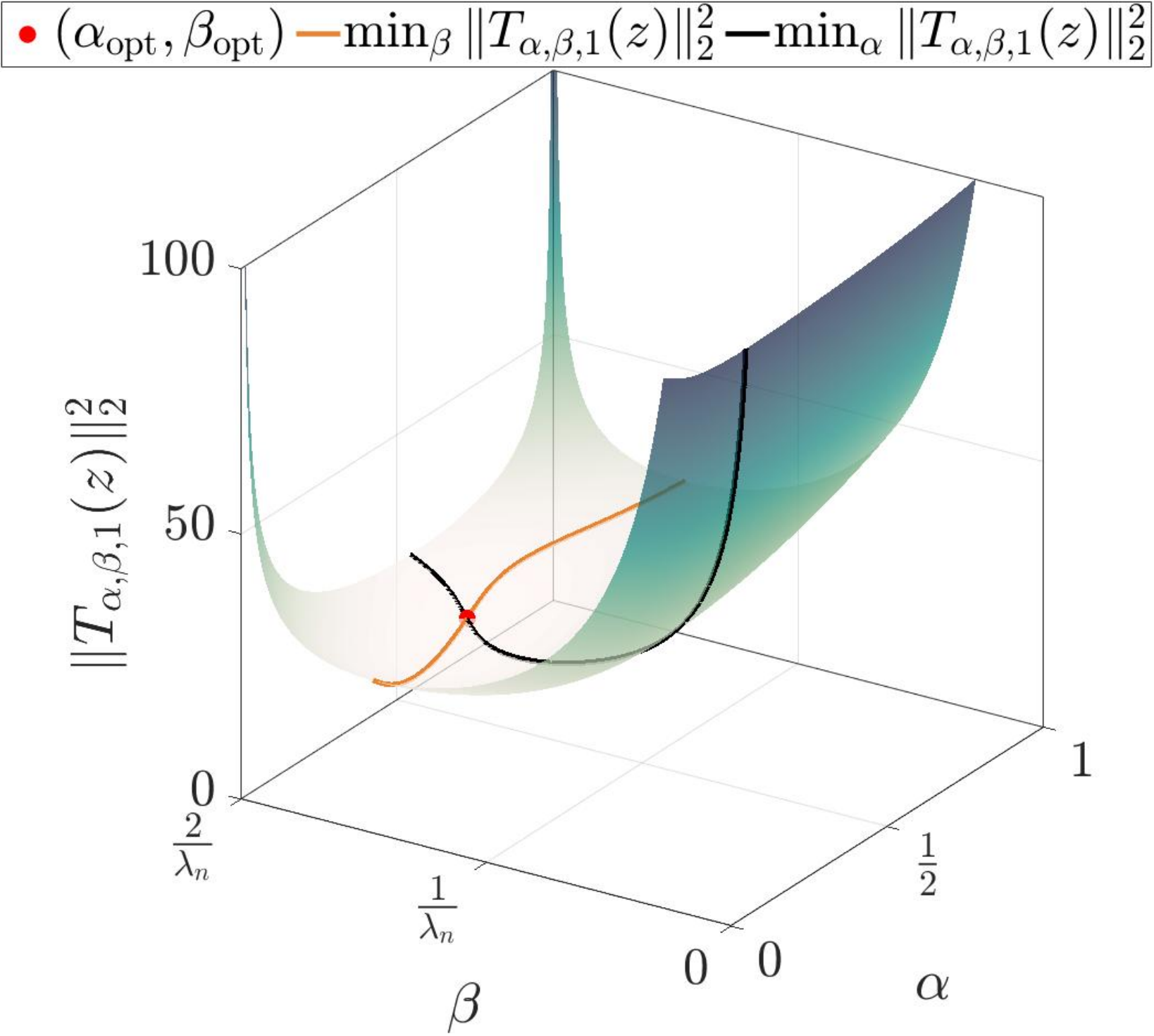}
	\caption{Surface of $\|T_{\alpha,\beta,1}(z)\|_2^2$ on $C^2_{20}$.}
	\label{Fig4}
\end{figure}

{
\subsection{Tests on Scale-free Networks}
Aforementioned discussions focus on networks with symmetric geometric structures,
but they often fail to capture the complexity and heterogeneity of the real-world connections.
We now tests our results on scale-free networks \cite{Albert-LaszloBarabasi2009},
which are widespread in the World Wide Web, social networks, power grids, etc.
These networks can be characterized by a power-law degree distribution,
where a few hubs possess an exceptionally large number of connections, while most nodes have only a few links.
We generate a scale-free network $\mathcal{G}_{\mathrm{sf}}^{50}$ with $50$ vertices via the Barab{\'a}si-Albert model \cite{Albert-LaszloBarabasi1999},
 which is shown in Fig. \ref{Fig5-1}.
The degree of each node is visually encoded through both color and size; specifically, larger and lighter nodes denote higher degrees.
For the scale-free network, Fig. \ref{Fig5-2} illustrates the trends of optimal memory depths with respect to $(\alpha,\beta)$ and the corresponding surfaces of $H_2$ performance metrics.
The global minimizers $(\alpha_{\mathrm{opt}},\beta_{\mathrm{opt}},\theta_{\mathrm{opt}})$, indicated by a red circle, occurs at the intersection of the gray (minimized over $\beta$) and black (minimized over $\alpha$) curves.

\begin{figure}[!ht]
	\centering
	\begin{subfigure}{0.35\textwidth}
			\centering
			\includegraphics[width=\linewidth]{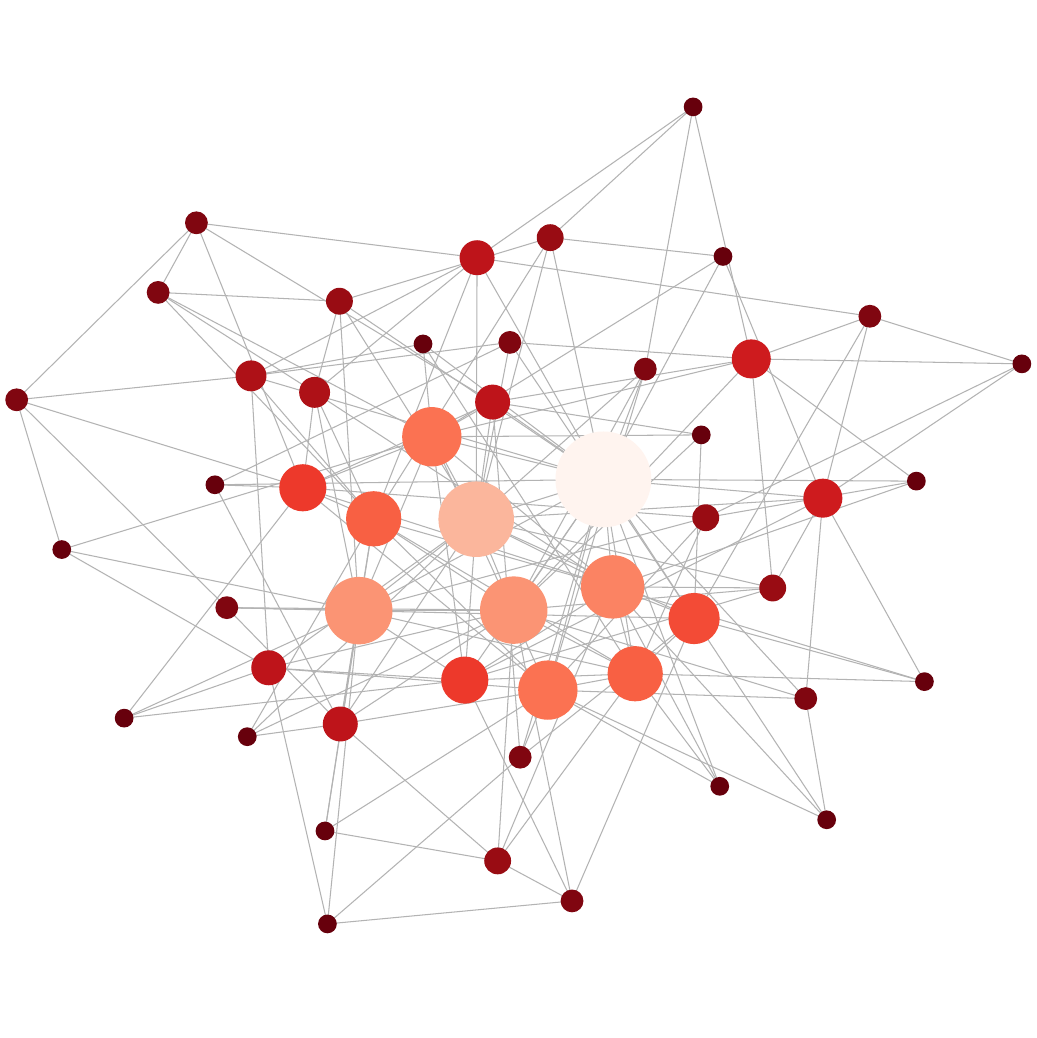}
			\caption{}
			\label{Fig5-1}
	\end{subfigure}
	\begin{subfigure}{0.35\textwidth}
			\centering
			\includegraphics[width=\linewidth]{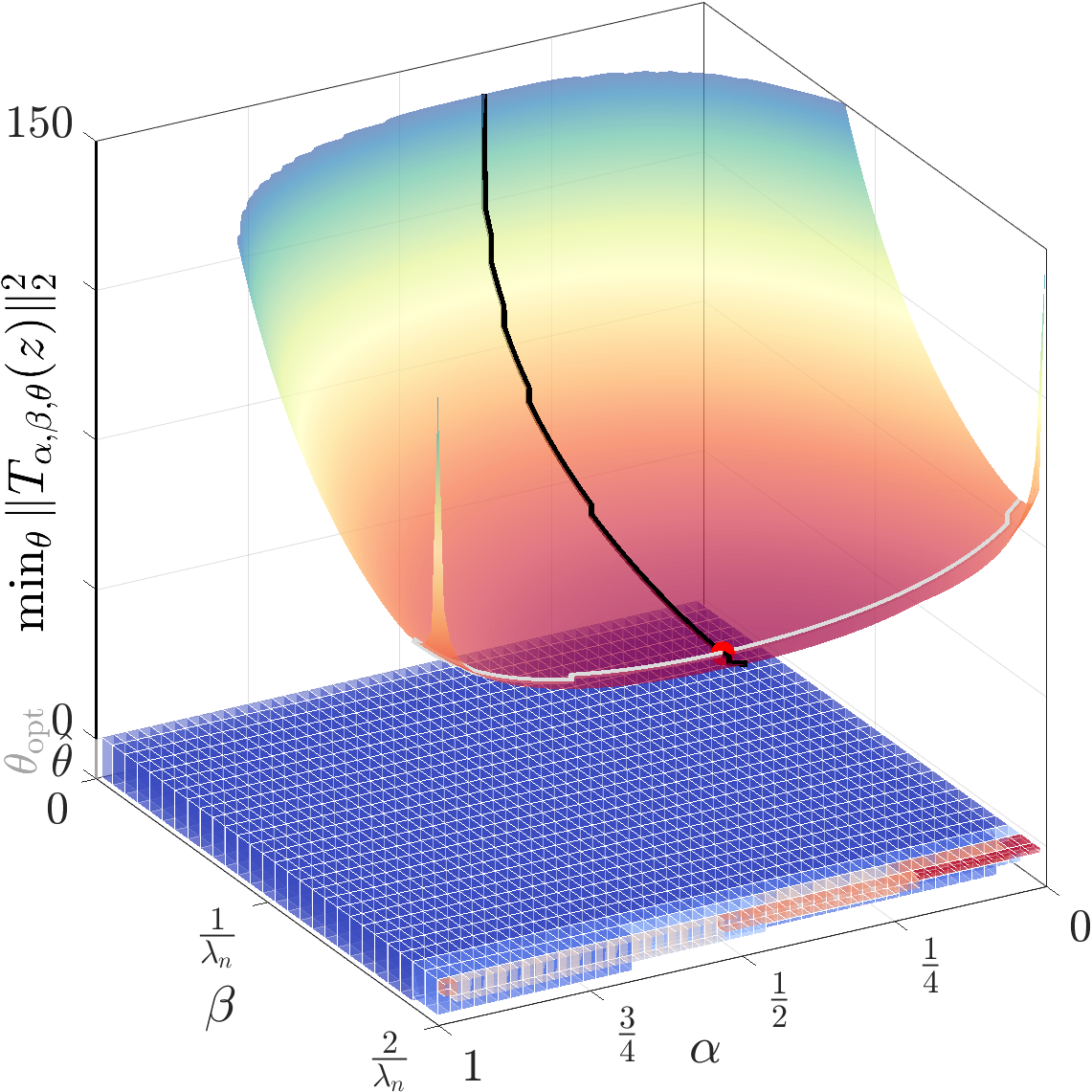}
			\caption{}
			\label{Fig5-2}
	\end{subfigure}
	\caption{Experimental results on the scale-free network $\mathcal{G}_{\mathrm{sf}}^{50}$.}
	\label{Fig5}
\end{figure}

Supported by the experimental results and extensive tests not shown here,
we find that the transition of optimal memory depth from long-term to short-term memory, with increasing coupling gain, is exclusive to denser scale-free networks.
In contrast, for sparser networks, long-term memory is optimal.
These tests also confirm the uniqueness of the optimal parameter pair $(\alpha_{\mathrm{opt}},\beta_{\mathrm{opt}},\theta_{\mathrm{opt}})$ under finite memory depths.
}
\section{Conclusion}\label{Sec:6}
This paper studied the role of memory information, particularly memory depth, on the $H_2$ performance of multi-agent consensus networks.
The memory depth refers to the temporal extent to which the agent invoke its past iteration state.
We proposed a memory-based consensus protocol in the form of linear extrapolation,
which employs past single-step states from agents and their neighbors, with tunable memory depth.
A necessary and sufficient condition for achieving consensus was given,
and we discovered an useful property of the proposed protocol: consensus attained under remote memory ensure consensus under recent memory.
Furthermore, we developed a universal approach to establish the quantitative relation between the $H_2$ performance metric, memory factor, memory depth, coupling gain, and Laplacian spectrum.
This leads to a more explicit analytical expression for the $H_2$ performance metric,
characterized by continued fractions in terms of the memory factor and coupling gain.
For the case of $\alpha=\frac{1}{2}$,
we proved that memory information at any accessible depth enhances $H_2$ performance,
and determined the optimal memory depth within specific parameter regions.
Finally, we discussed the direct correlation between the optimal memory depth with topological structure, group size, and network density;
trends of $H_2$ performance metric;
transitions from long-term to short-term memory;
the superiority of hybrid information strategies;
and the marginal utility of increasing memory depth.
These findings deserve to be further studied in future works.
\appendices
\section{Proof of Lemma \ref{lem:cfrac_property}}\label{app:pro_cfrac}
For $\mathscr{G}_n(\tau)$ and $\tau\in(0,1)$, combining $\mathscr{G}_1(\tau)>1$ with the following recursive form
\begin{equation}\label{eq:Appendix_A_2}
	\mathscr{G}_{n+1}(\tau)=\frac{2}{\tau}-\frac{1}{\mathscr{G}_n(\tau)},
\end{equation}
it can be verified that $\mathscr{G}_{n}(\tau)>1$, $\forall n\in\mathbb{N}^+$.
Similarly, we can demonstrate that $\mathscr{G}_{n}(\tau)<-1$ $\forall n\in\mathbb{N}^+$ when $\tau\in(-1,0)$.
Then, it follows from the facts
\begin{equation*}
	\mathscr{G}_{n+1}(\tau)-\mathscr{G}_{n}(\tau)
	=\frac{\mathscr{G}_{2}(\tau)-\mathscr{G}_{1}(\tau)}{\prod_{m=1}^{n-1}[\mathscr{G}_{m+1}(\tau)\mathscr{G}_{m}(\tau)]}
\end{equation*}
and $\mathscr{G}_{m+1}(\tau)\mathscr{G}_{m}(\tau)>0$, $\forall m\in\mathbb{N}^+$
that the sign of $\mathscr{G}_{n+1}(\tau)-\mathscr{G}_{n}(\tau)$ is ultimately determined by the sign of $\mathscr{G}_{2}(\tau)-\mathscr{G}_{1}(\tau)=\tau-\frac{\tau}{2-\tau^2}$.
Note that $\tau-\frac{\tau}{2-\tau^2}>0$ when $\tau\in(0,1)$ and $\tau-\frac{\tau}{2-\tau^2}<0$ when $\tau\in(-1,0)$.
Therefore, \eqref{eq:Gn} is proved.

For $\mathscr{F}_n(\tau)$, we can get the following recursive form
\begin{equation}\label{eq:Appendix_A_1}
	\mathscr{F}_{n+1}(\tau)=\frac{2}{\tau}-\frac{1}{\mathscr{F}_n(\tau)}.
\end{equation}
If $\tau\in(0,1)$, it is obvious that $\mathscr{F}_{n+1}(\tau)>1$ as long as $\mathscr{F}_{n}(\tau)>1$.
Since $\mathscr{F}_{1}(\tau)>1$,
we have $\mathscr{F}_{n}(\tau)>1$ $\forall n\in\mathbb{N}^+$.
Similarly, we can prove that $\mathscr{F}_{n}(\tau)<-1$ $\forall n\in\mathbb{N}^+$ when $\tau\in(-1,0)$.
It follows from \eqref{eq:Appendix_A_2} and \eqref{eq:Appendix_A_1} that
\begin{gather*}
	\mathscr{F}_{n+1}(\tau)-\mathscr{F}_{n}(\tau)
	=\frac{\mathscr{F}_{2}(\tau)-\mathscr{F}_{1}(\tau)}{\prod_{m=1}^{n-1}[\mathscr{F}_{m+1}(\tau)\mathscr{F}_{m}(\tau)]}, \\
	\mathscr{G}_{n}(\tau)-\mathscr{F}_{n}(\tau)
	=\frac{\mathscr{G}_{1}(\tau)-\mathscr{F}_{1}(\tau)}{\prod_{m=1}^{n-1}[\mathscr{G}_{m}(\tau)\mathscr{F}_{m}(\tau)]}, \\
	\mathscr{F}_{n+1}(\tau)-\mathscr{G}_{n}(\tau)
	=\frac{\mathscr{F}_{2}(\tau)-\mathscr{G}_{1}(\tau)}{\prod_{m=1}^{n-1}[\mathscr{F}_{m+1}(\tau)\mathscr{G}_{m}(\tau)]}.
\end{gather*}

Note that $\mathscr{F}_{2}(\tau)-\mathscr{F}_{1}(\tau)=1-\frac{\tau}{2-\tau}>0$ and
$\mathscr{G}_{1}(\tau)-\mathscr{F}_{1}(\tau)=1-\tau>0$ hold for all $\tau\in(-1,0)\cup(0,1)$.
Furthermore, $\mathscr{F}_{2}(\tau)-\mathscr{G}_{1}(\tau)=\tau-\frac{\tau}{2-\tau}>0$ when $\tau\in(0,1)$ and $\tau-\frac{\tau}{2-\tau}<0$ when $\tau\in(-1,0)$.
Combining with the fact that $\mathscr{F}_{n}(\tau)$ and $\mathscr{G}_{n}(\tau)$ share the same sign,
\eqref{eq:Fn} is obtained.
\hfill\qedsymbol

\section{Proof of Lemma \ref{lem:consensus_condition}}\label{app:pro_consensus_condition}
Recall that $\mathcal{G}$ is connected,
hence $0=\lambda_1<\lambda_2\leq\cdots\leq\lambda_n$.
It follows from
$\mathscr{P}(\theta,1,\lambda_{1})=0$,
$\frac{\mathrm{d}\mathscr{P}(\theta,\gamma,\lambda_{1})}{\mathrm{d}\gamma}\big|_{\gamma=1}\neq0$,
and
$\mathscr{P}(\theta,1,\lambda_{i})\neq0$ $\forall i\in\mathcal{I}_{2,n}$
that $1$ is an algebraically simple eigenvalue of $\Theta$,
which is denoted as $\gamma_{1,1}=1$.

For other eigenvalues $\gamma_{1,2},\dots,\gamma_{1,\theta+1}$,
assume that at least one eigenvalue
$\gamma_{1,k}$ ($k\in\mathcal{I}_{2,\theta+1}$)
is not inside the unit circle, i.e., $|\gamma_{1,k}|\geq1$.
If $|\gamma_{1,k}|>1$,
then we have 
$|\gamma_{1,k}^{\theta+1}-\gamma_{1,k}^{\theta}\alpha|
=|\gamma_{1,k}|^{\theta}\cdot |\gamma_{1,k}-\alpha|
\geq|\gamma_{1,k}|^{\theta}(|\gamma_{1,k}|-\alpha)
>1-\alpha$,
which contradicts the fact
$\mathscr{P}(\theta,\gamma_{1,k},\lambda_{1})=0$.
If $|\gamma_{1,k}|=1$, 
then
$|\gamma_{1,k}^{\theta+1}-\gamma_{1,k}^{\theta}\alpha|
=|\gamma_{1,k}-\alpha|=1-\alpha$
holds if and only if $\gamma_{1,k}=1$.
However, this contradicts the uniqueness of the eigenvalue $\gamma_{1,1}=1$.
Thus, $\gamma_{1,2},\dots,\gamma_{1,\theta+1}$ are all within the unit circle in the complex plane.

Then, the remaining proof of this lemma is essentially similar to that of Theorem 1 in \cite{FengXiao2006}.
According to this theorem,
we know that the noise-free multi-agent network \eqref{eq:DT_dynamics} achieves consensus if and only if $1$ is an algebraically simple eigenvalue of $\Theta$ and all the other eigenvalues lie within the unit circle in the complex plane.
Therefore, combining with aforementioned properties of $\gamma_{1,k}$,
the noise-free multi-agent network \eqref{eq:DT_dynamics} achieves consensus under protocol \eqref{eq:memory_protocol} if and only if $|\gamma_{i,k}|<1$ $\forall i\in\mathcal{I}_{2,n}, k\in\mathcal{I}_{1,\theta+1}$.

\section{Proof of Theorem \ref{thm:consensus_property}}\label{app:pro_consensus_property}
Lemma \ref{lem:consensus_condition} implies that the consensus is achieved under $(\theta+1)$-depth memory if and only if the trinomials
$\mathscr{P}(\theta+1,\gamma,\lambda_{i})$
are Schur stable for all $i\in\mathcal{I}_{2,n}$.
Then, we only need to prove that
$\mathscr{P}(\theta,\gamma,\lambda_{i})$
is Schur stable as long as
$\mathscr{P}(\theta+1,\gamma,\lambda_{i})$
is Schur stable.

According to the Jury stability criterion \cite{KatsuhikoOgata1995},
we construct the Jury stability table of the trinomial
$\mathscr{P}(\theta+1,\gamma,\lambda_{i})$ in Table \ref{tab:I-a},
where the elements in the first row is composed of the coefficients in $\mathscr{P}(\theta+1,\gamma,\lambda_{i})$ arranged in the ascending order with respect to the powers of $\gamma$, i.e.,
\begin{equation*}
	a_0=-(1-\alpha)(1-\beta\lambda_{i}),a_{\theta+1}=-\alpha(1-\beta\lambda_{i}),a_{\theta+2}=1.
\end{equation*}

\begin{table*}[!ht]\caption{The Jury stability tables for two trinomials.}
	\centering
	\subfloat[$\mathscr{P}(\theta+1,\gamma,\lambda_{i})$]{
		\begin{tblr}
			{colspec=
				{Q[2.45em,c]
					|
					Q[wd=1.7em,c]
					Q[wd=1.7em,c]
					Q[wd=1.7em,c]
					Q[wd=1.7em,c]
					Q[wd=1.7em,c]
					Q[wd=1.7em,c]
					Q[wd=1.7em,c]
					Q[wd=1.7em,c]
					Q[wd=1.7em,c]}
			}
			\hline[1pt]
			Row & $\gamma^0$ & $\gamma^1$ & $\gamma^2$ & $\gamma^3$ & $\cdots$ & $\gamma^{\theta-1}$ & $\gamma^{\theta}$ & $\gamma^{\theta+1}$ & $\gamma^{\theta+2}$ \\
			\hline
			$1$ & \SetCell{bg=brown9}$a_0$ & $0$ & $0$ & $0$ & $\cdots$ & $0$ & $0$ & \SetCell{bg=brown9}$a_{\theta+1}$ & \SetCell{bg=brown9}$a_{\theta+2}$ \\
			$2$ & \SetCell{bg=brown9}$a_{\theta+2}$ & \SetCell{bg=brown9}$a_{\theta+1}$ & $0$ & $0$ & $\cdots$ & $0$ & $0$ & $0$ & \SetCell{bg=brown9}$a_0$ \\
			$3$ & \SetCell{bg=teal9}$b_0$ & \SetCell{bg=teal9}$b_1$ & $0$ & $0$ & $\cdots$ & $0$ & $0$ & \SetCell{bg=teal9}$b_{\theta+1}$ &  \\
			$4$ & \SetCell{bg=teal9}$b_{\theta+1}$ & $0$ & $0$ & $0$ & $\cdots$ & $0$ & \SetCell{bg=teal9}$b_1$ & \SetCell{bg=teal9}$b_0$ & \\
			$5$ & \SetCell{bg=azure9}$c_0$ & \SetCell{bg=azure9}$c_1$ & $0$ & $0$ & $\cdots$ & $0$ & \SetCell{bg=azure9}$c_{\theta}$ &  & \\
			$6$ & \SetCell{bg=azure9}$c_{\theta}$ & $0$ & $0$ & $0$ & $\cdots$ & \SetCell{bg=azure9}$c_1$ & \SetCell{bg=azure9}$c_0$ & & \\
			$\vdots$ & $\vdots$ & $\vdots$ & $\vdots$ & $\vdots$ & & & & & \\
			$2\theta-1$ & \SetCell{bg=blue9}$p_0$ & \SetCell{bg=blue9}$p_1$ & $0$ & \SetCell{bg=blue9}$p_3$ &  &  &  &  & \\
			$2\theta$ & \SetCell{bg=blue9}$p_3$ & $0$ & \SetCell{bg=blue9}$p_1$ & \SetCell{bg=blue9}$p_0$ &  &  &  &  & \\
			$2\theta+1$ & \SetCell{bg=purple9}$q_0$ & \SetCell{bg=purple9}$q_1$ & \SetCell{bg=purple9}$q_2$ &  &  &  &  &  & \\
			$2\theta+2$ & \SetCell{bg=purple9}$q_2$ & \SetCell{bg=purple9}$q_1$ & \SetCell{bg=purple9}$q_0$ &  &  &  &  &  & \\
			\hline[1pt]
		\end{tblr}\label{tab:I-a}
	}
	\subfloat[$\mathscr{P}(\theta,\gamma,\lambda_{i})$]{
		\begin{tblr}
			{colspec=
				{Q[2.45em,c]
					|
					Q[wd=1.7em,c]
					Q[wd=1.7em,c]
					Q[wd=1.7em,c]
					Q[wd=1.7em,c]
					Q[wd=1.7em,c]
					Q[wd=1.7em,c]
					Q[wd=1.7em,c]
					Q[wd=1.7em,c]
					Q[wd=1.7em,c]}
			}
			\hline[1pt]
			Row & $\gamma^0$ & $\gamma^1$ & $\gamma^2$ & $\cdots$ & $\gamma^{\theta-2}$ & $\gamma^{\theta-1}$ & $\gamma^{\theta}$ & $\gamma^{\theta+1}$ \\
			\hline
			$1$ & \SetCell{bg=brown9}$a_0$ & $0$ & $0$ & $\cdots$ & $0$ & $0$ & \SetCell{bg=brown9}$a_{\theta+1}$ & \SetCell{bg=brown9}$a_{\theta+2}$ \\
			$2$ & \SetCell{bg=brown9}$a_{\theta+2}$ & \SetCell{bg=brown9}$a_{\theta+1}$ & $0$ & $\cdots$ & $0$ & $0$ & $0$ & \SetCell{bg=brown9}$a_0$ \\
			$3$ & \SetCell{bg=teal9}$b_0$ & \SetCell{bg=teal9}$b_1$ & $0$ & $\cdots$ & $0$ & $0$ & \SetCell{bg=teal9}$b_{\theta+1}$ &  \\
			$4$ & \SetCell{bg=teal9}$b_{\theta+1}$ & $0$ & $0$ & $\cdots$ & $0$ & \SetCell{bg=teal9}$b_1$ & \SetCell{bg=teal9}$b_0$ & \\
			$5$ & \SetCell{bg=azure9}$c_0$ & \SetCell{bg=azure9}$c_1$ & $0$ & $\cdots$ & $0$ & \SetCell{bg=azure9}$c_{\theta}$ &  & \\
			$6$ & \SetCell{bg=azure9}$c_{\theta}$ & $0$ & $0$ & $\cdots$ & \SetCell{bg=azure9}$c_1$ & \SetCell{bg=azure9}$c_0$ & & \\
			$\vdots$ & $\vdots$ & $\vdots$ & $\vdots$ & & & & & \\
			$2\theta-1$ & \SetCell{bg=blue9}$p_0$ & \SetCell{bg=blue9}$p_1$ &  \SetCell{bg=blue9}$p_3$ &  &  &  &  & \\
			$2\theta$ & \SetCell{bg=blue9}$p_3$ & \SetCell{bg=blue9}$p_1$ & \SetCell{bg=blue9}$p_0$ &  &  &  &  & \\
			\hline[1pt]
		\end{tblr}\label{tab:I-b}
	}
\end{table*}

For each $k\in\mathcal{I}_{0,\theta}$, the $(2k+2)$-th row consists of the elements of the $(2k+1)$-th row arranged in reverse order,
and only the potentially nonzero elements are labeled and colored,
which are given by
\begin{gather*}
	b_0=a_0^2-a_{\theta+2}^2,
	b_1=-a_{\theta+1}a_{\theta+2},
	b_{\theta+1}=a_0a_{\theta+1},\\
	c_0=b_0^2-b_{\theta+1}^2,
	c_1=b_0b_1,
	c_{\theta}=-b_1b_{\theta+1},\\
	\vdots\\
	q_0=p_0^2-p_3^2,
	q_1=p_0p_1,
	q_2=-p_1p_3.
\end{gather*}
Then, $\mathscr{P}(\theta+1,\gamma,\lambda_{i})$ is Schur stable if and only if
$\mathscr{P}(\theta+1,1,\lambda_{i})>0$,
\begin{equation}\label{eq:Jury_condition_1}
	\mathscr{P}(\theta+1,-1,\lambda_{i})
	\begin{cases}
		>0, & \mathrm{if}~\theta~\mathrm{is\ even},  \\
		<0, & \mathrm{if}~\theta~\mathrm{is\ odd},
	\end{cases}
\end{equation}
and the following stability conditions are all satisfied:
\begin{equation}\label{eq:Jury_condition_2}
	\begin{gathered}
		|a_0|<a_{\theta+2},|b_0|>|b_{\theta+1}|,|c_0|>|c_{\theta}|,\\
		\vdots\\
		|p_0|>|p_3|,|q_0|>|q_2|.
	\end{gathered}
\end{equation}
Similarly, the Jury stability table of the trinomial
$\mathscr{P}(\theta,\gamma,\lambda_{i})$ can be given by Table \ref{tab:I-b},
which contains exactly the same elements as those in Table \ref{tab:I-a}, excluding $q_0$, $q_1$, and $q_2$.

Then, $\mathscr{P}(\theta,\gamma,\lambda_{i})$ is Schur stable if and only if $\mathscr{P}(\theta,1,\lambda_{i})>0$,
\begin{equation}\label{eq:Jury_condition_3}
	\mathscr{P}(\theta,-1,\lambda_{i})
	\begin{cases}
		>0, & \mathrm{if}~\theta~\mathrm{is\ odd},  \\
		<0, & \mathrm{if}~\theta~\mathrm{is\ even},
	\end{cases}
\end{equation}
and all conditions in \eqref{eq:Jury_condition_2}, except $|q_0|>|q_2|$, are satisfied.

It follows from $\beta\in(0,\frac{2}{\lambda_{n}})$ that
$\mathscr{P}(\theta,1,\lambda_{i})=\beta\lambda_{i}>0$, \eqref{eq:Jury_condition_1}, and \eqref{eq:Jury_condition_3} hold for all $i\in\mathcal{I}_{2,n}$ and $\theta\in\mathbb{N}^+$.
Furthermore, one can observe that the potentially nonzero elements in Table \ref{tab:I-b} must satisfy the stability conditions of $\mathscr{P}(\theta,1,\lambda_{i})$ if \eqref{eq:Jury_condition_2} holds.
Therefore, for all $i\in\mathcal{I}_{2,n}$ and $\theta\in\mathbb{N}^+$,
if $\mathscr{P}(\theta+1,\gamma,\lambda_{i})$ is Schur stable,
then $\mathscr{P}(\theta,\gamma,\lambda_{i})$ is also Schur stable.

The proof is completed. \hfill\qedsymbol

\section{Proof of Lemma \ref{lem:H2_property}}\label{app:pro_H2_property}
	Define a block matrix
\begin{equation*}
	\mathcal{P}=
	\begin{bmatrix}
		\mathcal{P}_{1,1} & \dots & \mathcal{P}_{1,\theta+1}\\
		\vdots & \ddots & \vdots \\
		\mathcal{P}_{\theta+1,1} & \dots & \mathcal{P}_{\theta+1,\theta+1}
	\end{bmatrix}\in\mathbb{R}^{h\times h},
\end{equation*}
where each block is an $n\times n$ matrix,
such that $\mathcal{P}=(zI_h-\tilde{\Theta})^{-1}$.
It is inferred from $(zI_h-\tilde{\Theta})\mathcal{P}=I_h$ that
\begin{equation}\label{eq:P_property_1}
	z\mathcal{P}_{i,\theta+1}-\mathcal{P}_{i+1,\theta+1}=\mathbf{0}\ \forall i\in\mathcal{I}_{1,\theta}
\end{equation}
and
\begin{equation}\label{eq:P_property_2}
	(\alpha-1)\bar{\Phi}\mathcal{P}_{1,\theta+1}+(zI_n-\alpha\bar{\Phi})\mathcal{P}_{\theta+1,\theta+1}=I_n.
\end{equation}
By inspecting \eqref{eq:P_property_1},
we can derive that $\mathcal{P}_{\theta+1,\theta+1}=z\mathcal{P}_{\theta,\theta+1}=\cdots=z^\theta \mathcal{P}_{1,\theta+1}$.
Combining with \eqref{eq:P_property_2},
we can get
$\mathcal{P}_{\theta+1,\theta+1}=z^\theta[(\alpha-1)\bar{\Phi}+z^\theta(zI_n-\alpha\bar{\Phi})]^{-1}$.
Then, we have
\begin{align*}
	T_1(z)&=
	\begin{bmatrix}
		\mathbf{0} & \bar{\Psi}
	\end{bmatrix}
	\mathcal{P}
	\begin{bmatrix}
		\mathbf{0}\\
		I_n
	\end{bmatrix} \\
	&=\bar{\Psi}\mathcal{P}_{\theta+1,\theta+1} \\
	&=\begin{bmatrix}
		0 & \mathbf{0}_n^\top \\
		\mathbf{0}_n & z^\theta[(\alpha-1)\Phi+z^\theta(zI_n-\alpha\Phi)]^{-1}
	\end{bmatrix}.
\end{align*}
Similarly, the transfer matrix of system \eqref{eq:augmented_system_3} can be given by
\begin{equation*}
	T_2(z)=z^\theta[(\alpha-1)\Phi+z^\theta(zI_n-\alpha\Phi)]^{-1}.
\end{equation*}
It can be readily obtained that
\begin{equation}\label{eq:P_property_3}
	T_1(z)
	=
	\begin{bmatrix}
		0 & \mathbf{0}_n^\top \\
		\mathbf{0}_n & T_2(z)
	\end{bmatrix}
	=
	Q^\top T_{\alpha,\beta,\theta}(z)Q.
\end{equation}

We can verify that $\hat{\Theta}$ has the same eigenvalues of $\Theta$ except $\gamma_{1,1},\dots,\gamma_{1,\theta+1}$.
As $\mathcal{G}$ is connected and $(\alpha,\beta)\in\Omega_\theta$,
system \eqref{eq:augmented_system_3} is asymptotically stable,
which means that $\|T_2(z)\|_2$ is well-defined.
Then, it follows from \eqref{eq:P_property_3} and Definition \ref{def:H2_norm} that $\|T_{\alpha,\beta,\theta}(z)\|_2=\|T_{1}(z)\|_2=\|T_{2}(z)\|_2$.
Therefore, $\|T_{\alpha,\beta,\theta}(z)\|_2$ is also well-defined and is completely determined by the asymptotically stable subsystem \eqref{eq:augmented_system_3}.\hfill\qedsymbol






\ifCLASSOPTIONcaptionsoff
  \newpage
\fi

\footnotesize 
\bibliography{references}

\end{document}